\documentclass[a4paper,reqno]{amsart}

\usepackage{amssymb,amsmath,amsthm,amsfonts}
\usepackage{graphicx,setspace}
\usepackage{color,comment}

\usepackage{hyperref}

\usepackage[left=3.0cm,right=3.0cm,top=3.0cm,bottom=3.0cm]{geometry}

\newcommand{\hamiltonian}{{\bf h}}
\newcommand{\potential}{{\bf v}}
\newcommand{\dispersion}{\mathfrak{e}}
\newcommand{\birman}{\mathfrak{b}}
\newcommand{\pmin}{p_0}
\newcommand{\pmax}{p^0}
\newcommand{\Tr}{{\ensuremath{\mathrm{Tr}}}}
\newcommand{\Ker}{{\ensuremath{\mathrm{Ker\,}}}}
\newcommand{\sign}{{\ensuremath{\mathrm{sign\,}}}}
\newcommand{\disc}{{\ensuremath{\mathrm{disc}}}}
\newcommand{\ess}{{\ensuremath{\mathrm{ess}}}}
\newcommand{\emin}{\dispersion_{\min}}
\newcommand{\emax}{\dispersion_{\max}}
\renewcommand{\d}{\mathrm{d}}

\renewcommand{\hat}{\widehat}
\renewcommand{\tilde}{\widetilde}

\newcommand{\C}{\mathbb{C}}

\newcommand{\R}{\mathbb{R}}
\newcommand{\T}{\mathbb{T}}
\newcommand{\Z}{\mathbb{Z}}

\newcommand{\cF}{\mathcal{F}}
\newcommand{\cN}{\mathcal{N}}

\numberwithin{equation}{section}

\theoremstyle{plain}
\newtheorem{theorem}{Theorem}[section]
\newtheorem{proposition}[theorem]{Proposition}
\newtheorem{lemma}[theorem]{Lemma}
\newtheorem{corollary}[theorem]{Corollary}
\newtheorem{hypothesis}[theorem]{Hypothesis}

\theoremstyle{definition}

\date{\today}

\begin{document}

\title[]{Bound states of discrete Schr\"odinger operators on one and two dimensional lattices}

\author[Sh. Kholmatov, S.N. Lakaev, F. Almuratov]{Sh. Kholmatov, S.N. Lakaev, F. Almuratov}

\address[Sh. Kholmatov]{University of Vienna,
Oskar-Morgenstern-Platz 1, 1090  Vienna, Austria}
\email{shokhrukh.kholmatov@univie.ac.at}

\address[S. Lakaev]{Samarkand State University, University boulevard 15, 140104 Samarkand, Uzbekistan}
\email{slakaev@mail.ru}

\address[F. Almuratov]{Samarkand State University, University boulevard 15, 140104 Samarkand, Uzbekistan}
\email{almurotov93@mail.ru}

\begin{abstract}
We study the spectral properties of discrete Schr\"odinger operator
$$
\widehat \hamiltonian_\mu=\widehat \hamiltonian_0 +
\mu \widehat  \potential,\qquad \mu\ge0,
$$
associated to a one-particle system in $d$-dimensional lattice
$\Z^d, $ $d=1,2,$ where the non-perturbed operator $\hat \hamiltonian_0$ is a self-adjoint Laurent-Toeplitz-type operator generated by $\hat
\dispersion:\Z^d\to\C$ and the potential $\hat \potential$ is the
multiplication operator by $\hat v:\Z^d\to\R.$ Under certain
regularity assumption on $\hat \dispersion$ and a decay assumption
on $\hat v$, we establish the existence or non-existence and also
the finiteness of eigenvalues of $\hat\hamiltonian_\mu.$ Moreover,
in the case of existence we study the asymptotics of eigenvalues of
$\hat \hamiltonian_\mu$ as $\mu\searrow 0.$
\end{abstract}

\maketitle

\section{Introduction}

In \cite{Klaus:1977} Klaus studied the eigenvalues of the Schr\"odinger operator $-d^2/dx^2 + \lambda V$ for $\lambda>0$ and $V$ obeying
$$
\int_\R (1 + |x|) |V(x)|dx<\infty,
$$
extending the results of Simon in \cite{Simon:1976} in case of $d=1.$  Klaus showed that if $\int V(x)dx>0,$ then for small and positive $\lambda$ there is no bound state, and if $\int V(x)dx \le0,$ then there exists a bound state $E(\lambda)$ and it satisfies
$$
(-E(\lambda))^{1/2} = -\frac{\lambda}{2} \int V(x)dx - \frac{\lambda^2}{4} \int V(x)|x-y|V(y)dxdy +o(\lambda^2)
$$
as $\lambda\searrow0.$

In the present paper we replace the Euclidean $d$-dimensional space $\R^d$ by the $d$-dimensional lattice $\Z^d,$ $d=1,2,$ and study the discrete spectrum of a large class of lattice Schr\"odinger operators $\hat \hamiltonian_\mu$ in $\ell^2(\Z^d)$ given by
$$
\widehat \hamiltonian_\mu:=\widehat \hamiltonian_0 +\mu \widehat  \potential,\qquad \mu\ge0,
$$
where the non-perturbed operator  $\hat \hamiltonian_0$ is a Laurent-Toeplitz-type operator with a generating function $\hat\dispersion\in \ell^1(\Z^d)$ satisfying $\overline{\hat \dispersion (-x)} = \hat\dispersion(x):$
$$
\hat\hamiltonian_0 f(x) = \sum\limits_{y\in\Z^d} \hat \dispersion(y) \hat f(x+y),\qquad \hat f\in \ell^2(\Z^d),
$$
and the potential $\hat \potential$ is the multiplication operator by a real-valued function $\hat v:\Z^d\to\R$ vanishing at infinity.

We also work with the representation of $\hamiltonian_\mu$ in $L^2(\T^d),$ where $\T^d=(-\pi,\pi]^d$ is the $d$-dimensional torus, the dual group of $\Z^d,$ equipped with the normalized Haar measure $\d p,$ i.e. $\int_{\T^d}\d p=1.$  The so-called ``momentum-space representation'' of $\hat \hamiltonian_0$ and $\potential$ are defined via the standard Fourier transform
$$
\cF:\ell^2(\Z^d) \to L^2(\T^d),\qquad \cF \hat f(p) = \sum\limits_{x\in\Z^d} \hat f(x)e^{ix\cdot p}
$$
as
$$
\hamiltonian_0:=\cF\hat \hamiltonian_0\cF^*\qquad \text{and}\qquad \potential:=\cF\hat\potential \cF^*,
$$
where
$$
\cF^*:L^2(\T^d) \to \ell^2(\Z^d),\qquad \cF^* f(x) = \int_{\T^d} f(p)e^{-ix\cdot p}\, \d p
$$
is the inverse Fourier transform. Then $\hat\hamiltonian_\mu$ is unitarily equivalent to the operator
$$
\hamiltonian_\mu:L^2(\T^d) \to L^2(\T^d),\qquad \hamiltonian_\mu:=\hamiltonian_0+\mu\potential.
$$

Note that $\hamiltonian_0$ is the multiplication operator by the continuous function $\dispersion:=\cF\hat \dispersion\in C(\T^d)$  and the potential $\potential$ is a convolution-type integral operator
$$
\potential f(p) = \int_{\T^d} v(p-q)f(q)\,\d q
$$
with the kernel distribution
$
v: = \cF\hat v.
$

Unless otherwise stated, throughout the paper we always assume that $\dispersion$ and $\hat v$ satisfy

\begin{hypothesis}\label{hyp:maina}
(a) The function $\dispersion$ has a unique minimum at $ \pmin $ and a unique maximum at $ \pmax ,$ the function $\dispersion$ is $C^2$ around $ \pmin $ and $ \pmax,$  and both  $ \pmin $ and $ \pmax $  are non-degenerate.

(b) there exists $\gamma\in(0,1)$ such that $0<\sum\limits_{x\in\Z^d} |x|^{2-d+\gamma}|\hat v(x)|<\infty.$
\end{hypothesis}

Note that since $\hat \potential$ is self-adjoint and compact, by the classical Weyl Theorem for any $\mu\ge0,$
\begin{equation}\label{essential_spectrum}
\sigma_\ess(\hat\hamiltonian_\mu) = \sigma(\hat \hamiltonian_0) = [\emin,\emax],
\end{equation}
where
$$
\emin:=\min \dispersion\qquad\text{and}\qquad \emax:=\max\dispersion.
$$

A typical example of $\hat \hamiltonian_0$ is the discrete Laplacian $\hat \Delta$ on $\Z^d,$ i.e.,
$$
\hat \Delta \hat f(x) :=\sum\limits_{j=1}^d \left[\hat f(x) - \frac{\hat f(x+1_j) + \hat f(x-1_j)}{2}\right],
$$
where $\{1_j\}$ is the basis of the lattice. In this case
$$
\dispersion(p) :=\sum\limits_{j=1}^d (1 - \cos q_j)
$$
satisfies Hypothesis \ref{hyp:maina} (a).

The main aim of the current paper is to study the discrete spectrum of $\hat \hamiltonian_\mu,$ in particular, the existence or non-existence, the uniqueness and finiteness of eigenvalues, and also the asymptotics of eigenvalues absorbed into
the essential spectrum as $\mu\searrow0.$

Our first result is related to the  Bargmann-type estimates for the number of eigenvalues of $\hat \hamiltonian_\mu$
outside the essential spectrum.

\begin{theorem}\label{teo:finiteness_eigenvalues}
Assume Hypothesis \ref{hyp:maina}. Then for any $\mu>0,$
\begin{equation}\label{posi_eigenvalues}
\cN^+(\hat\hamiltonian_\mu,\emax) \le 1+   C_1\,\mu \sum\limits_{x\in\Z^d} |x|^{2-d+\gamma} |\hat v(x)|
\end{equation}
and
\begin{equation}\label{nega_eigenvalues}
 \cN^-(\hat\hamiltonian_\mu,\emin) \le 1+ C_2\,\mu \sum\limits_{x\in\Z^d} |x|^{2-d+\gamma}|\hat v(x)|,
\end{equation}
where $C_1,C_2>0$ are coefficients depending only on $\dispersion$ and $\gamma.$  In particular, the number of eigenvalues of $\hat \hamiltonian_\mu$ outside the essential spectrum is finite for any $\mu>0.$
\end{theorem}

We note that Theorem \ref{teo:finiteness_eigenvalues} improves the upper bound for the number of eigenvalues obtained in \cite[Theorem 1.2]{BPL:2018}. To the best of our knowledge, estimates of the form \eqref{posi_eigenvalues}-\eqref{nega_eigenvalues} are known only for the discrete Laplacian. In fact, in $d=1$  sharp bounds for $\hat v$ for the finiteness of bound states of $-\hat \Delta + \hat \potential$ have been established in \cite{DT:2007.proc.ams} using some variational estimates. In $d=2$ an estimate of type \eqref{nega_eigenvalues} (with $\hat v\le0$ and with $\ln(1+|x|)$ in place of $|x|^\gamma$) for $-\hat \Delta + \hat \potential$ has been obtained in \cite{MV:2012.jms} applying Markov processes.   Analogous estimate in $\Z^2$ (again with $\hat v\le0$ and with $\ln(1+|x|)$ in place of $|x|^\gamma$) for $-\hat \Delta + \hat \potential$ has been obtained in \cite{RS:2013.aia} using some careful estimates for the two dimensional continuous Schr\"odinger operators together with interpolation arguments. In this paper we establish \eqref{posi_eigenvalues}-\eqref{nega_eigenvalues} without using those techniques, rather adapting the methods of Klaus in \cite{Klaus:1977}. Note that in the continuous case \eqref{posi_eigenvalues} does not make sense.

Our next results are related to the existence or
non-existence and also the uniqueness of eigenvalues of
$\hat\hamiltonian_\mu.$

\begin{theorem}\label{teo:existence_and_nonexistence}

Assume Hypothesis \ref{hyp:maina}. Then for any $\mu>0$:
\begin{itemize}
\item[(1)]  if $\sum\limits_{x\in\Z^d} \hat v(x)\ge0,$ then
$\sigma_\disc(\hat \hamiltonian_\mu) \cap (\emax,+\infty)\ne\emptyset;$

\item[(2)]  if $\sum\limits_{x\in\Z^d} \hat v(x)\le0,$ then
$\sigma_\disc(\hat \hamiltonian_\mu) \cap (-\infty,\emin)\ne\emptyset.$
\end{itemize}
Moreover, there exists $\mu_o:=\mu_o(\dispersion,\hat v)>0$ such that for any $\mu\in(0,\mu_o)$:
\begin{itemize}
\item[(a)]  if $\sum\limits_{x\in\Z^d} \hat v(x)>0,$ then $\sigma_\disc(\hat \hamiltonian_\mu) \cap (\emax,+\infty)$ is a singleton $\{E(\mu)\}$ and  $\sigma_\disc(\hat \hamiltonian_\mu) \cap (-\infty,\emin)=\emptyset;$

\item[(b)]  if $\sum\limits_{x\in\Z^d} \hat v(x)<0,$ then  $\sigma_\disc(\hat \hamiltonian_\mu) \cap (-\infty,\emin)$ is a singleton $\{e(\mu)\}$ and  $\sigma_\disc(\hat \hamiltonian_\mu) \cap (\emax,+\infty)=\emptyset;$

\item[(c)]  if $\sum\limits_{x\in\Z^d} \hat v(x)=0,$ then both $\sigma_\disc(\hat \hamiltonian_\mu) \cap (\emax,+\infty)$ and  $\sigma_\disc(\hat \hamiltonian_\mu) \cap (-\infty,\emin)$ are singletons $\{E(\mu)\}$ and  $\{e(\mu)\},$ respectively.

\end{itemize}
\end{theorem}

We remark that the existence of eigenvalues, i.e., assertions (1)-(2) of Theorem \ref{teo:existence_and_nonexistence} can also be obtained from \cite[Theorem 3.19]{HHRV:2017}, however, methods of \cite{HHRV:2017} seem not sufficient to establish the remaining assertions such as non-existence and uniqueness of eigenvalues.

Notice that by the linearity of $\mu\mapsto \hamiltonian_\mu,$ being a unique and isolated point of the discrete spectrum, both $\mu\in(0,\mu_o)\mapsto E(\mu)$ and $\mu\in(0,\mu_o)\mapsto e(\mu)$ are analytic. Moreover, $E(\mu)\searrow\emax$ and $e(\mu)\nearrow\emin$ as $\mu\searrow 0$ so that both eigenvalues are absorbed by the essential spectrum as $\mu\searrow0.$  Now we study their absorption rate.

\begin{theorem}\label{teo:asymptotics}
Assume Hypothesis \ref{hyp:maina} and additionally suppose that $\dispersion\in C^{3,\alpha}$ around $\pmax$ and $\pmin$ for some $\alpha\in(0,\gamma/8].$ Then there exists $\mu_1:=\mu_1(\dispersion,\hat v)\in(0,\mu_o)$ such that for any $\mu\in(0,\mu_1)$:
\begin{itemize}
\item[(a)]  if $\kappa_0:=\sum\limits_{x\in\Z^d} \hat v(x)>0,$ then
\begin{equation}\label{E_ning_asymp1}
E(\mu)- \dispersion_{\max}  =
\begin{cases}
\mu^2\Big[\kappa_0 a_1  + \mu^\alpha \Phi_1(\mu)\Big]^2 & \text{if $d=1,$}\\[4mm]
e^{-\frac{1}{\kappa_0 b_1\mu}}\,\Big[c_1 + \Psi_1(\mu)\Big]  & \text{if $d=2;$}
\end{cases}
\end{equation}
%

\item[(b)]  if $\kappa_0:=\sum\limits_{x\in\Z^d} \hat v(x)<0,$ then
\begin{equation}\label{e_ning_asymp1}
\emin - e(\mu) =
\begin{cases}
\mu^2 \Big[-\kappa_0a_2 + \mu^\alpha \Phi_2(\mu)\Big]^2 & \text{if $d=1,$}\\[4mm]
e^{\frac{1}{\kappa_0b_2\mu}}\,\Big[c_2 + \Psi_2(\mu)\Big] & \text{if $d=2;$}
\end{cases}
\end{equation}

\item[(c)]  if $\sum\limits_{x\in\Z^d} \hat v(x)=0,$ then both integrals
$$
\kappa_1:=\int_{\T^d} \frac{|v(p- \pmax )|^2\,\d p}{\dispersion_{\max}- \dispersion(p)},\qquad \kappa_2:=\int_{\T^d} \frac{|v(p- \pmin )|^2\,\d p}{\dispersion(p) - \dispersion_{\min}}
$$
are finite
and
\begin{equation}\label{E_ning_asymp01}
E(\mu)- \dispersion_{\max}  =
\begin{cases}
\mu^4\,\Big[\kappa_1a_3 + \mu^\gamma\ln^2\mu\,\Phi_3(\mu) \Big]^2 & \text{if $d=1,$}\\[4mm]
e^{- \frac{c_3}{\mu^2 \big(\sqrt{4\kappa_1 + b_3^2\mu^2} + b_3\mu\big)^2}}\,\Big[d_3+\Psi_3(\mu)\Big]\,
& \text{if $d=2,$}
\end{cases}
\end{equation}
and
\begin{equation}\label{e_ning_asymp01}
\dispersion_{\min} -e(\mu)  =
\begin{cases}
\mu^4\,\Big[\kappa_2a_4 + \mu^\gamma\ln^2\mu\,\Phi_4(\mu) \Big]^2 & \text{if $d=1,$}\\[4mm]
e^{- \frac{c_4}{\mu^2 \big(\sqrt{4\kappa_2 + b_4^2\mu^2} + b_4\mu\big)^2}}\,\Big[d_4+\Psi_4(\mu)\Big]\,
& \text{if $d=2.$}
\end{cases}
\end{equation}
\end{itemize}
Here $a_i,b_i,c_i,d_i>0$ are constants depending only on $\dispersion,$ and $\Phi_i,\Psi_i: [0,\mu_1]\to\R$ are continuous.
\end{theorem}

We remark that the asymptotics for $e(\mu)$ in $d=1$ corresponds to the continuous counterparts obtained in \cite{Klaus:1977,Simon:1976}, however, in $d=2$ the asymptotics  \eqref{e_ning_asymp1}-\eqref{e_ning_asymp01} sharper than the one in \cite[Theorem 3.4]{Simon:1976} obtained in the continuous case.

As in  \cite{Klaus:1977,Simon:1976} to prove Theorem \ref{teo:asymptotics} we obtain an asymptotic equation for $E(\mu)$ and $e(\mu).$ It turns out that in one dimensional case $E(\mu)$ and $e(\mu)$ satisfy
\begin{equation}\label{eq_for_eigen_case_d1}
\begin{aligned}
&\sqrt{E(\mu)-\emax} = [c_1+g_1(E(\mu)-\emax)]\mu^n,\\
&\sqrt{\emin-e(\mu)} = [c_2+g_2(\emin -e(\mu))]\mu^n,
\end{aligned}
\end{equation}
where $c_1,c_2>0$ are explicit constants, $g_1(z),g_2(z)\to0$ as $z\to0,$ and $n=1$ or $n=2$ depending on whether $\sum_x\hat v(x)$ is nonzero or zero. The equation \eqref{eq_for_eigen_case_d1} readily gives the first term of the asymptotics of $E(\mu)$ and $e(\mu).$ To identify the second term  we need to analyse the convergence rates of $g_1$ and $g_2.$  Similarly, in two dimensions the associated equations for $E(\mu)$ and $e(\mu)$ read as
\begin{equation}\label{eq_for_eigen_case_d2}
\begin{aligned}
&\frac{1}{-\mu^n\ln(E(\mu) - \emax)} =c_3+g_3(E(\mu)-\emax),\\
&\frac{1}{-\mu^n\ln(\emin -e(\mu)) } =c_4+g_4(\emin - e(\mu)),
\end{aligned}
\end{equation}
where $c_3,c_4>0$ are explicit constants, $g_3(z),g_4(z)\to0$ as $z\to0,$ and $n=1$ or $n=2$ depending on whether $\sum_x\hat v(x)$ is nonzero or zero. Unlike the one dimensional case, \eqref{eq_for_eigen_case_d2} is not sufficient even to find the first term of the asymptotics of $E(\mu)$ and $e(\mu),$ rather it yields  only  estimates of the form
$$
\exp(-(c_3 - \epsilon)^{-1}\mu^{-n}) \le E(\mu) -\emax \le \exp(-(c_3 + \epsilon)^{-1}\mu^{-n})
$$
and
$$
\exp(-(c_4 - \epsilon)^{-1}\mu^{-n}) \le \emin - e(\mu) \le \exp(-(c_4 + \epsilon)^{-1}\mu^{-n})
$$
for small $\epsilon>0$ (see e.g., \cite[Theorem 3.4]{Simon:1976}).

We prove  \eqref{E_ning_asymp1}-\eqref{e_ning_asymp01} obtaining careful estimates for $g_i$ using the perturbation theory for a (not necessarily self-adjoint) Birman-Schwinger operator $\birman(z),$ $z\in\R\setminus[\emin,\emax]$ (see Section \ref{sec:birman_sh_princip}). The equations for eigenvalues in \eqref{eq_for_eigen_case_d1} and \eqref{eq_for_eigen_case_d2} for the case $\sum_x \hat v(x)\ne0$ is obtained employing the similar arguments to \cite{Klaus:1977}. In this case $\birman(z)$ is represented as a small perturbation of rank-one operator (Lemma \ref{lem:residual_b_sh}) that has a unique non-zero eigenvalue. However, the case $\sum_x \hat v(x)=0$ requires more delicate analysis since in this case the previous perturbation-theory arguments fail. Here we were not able to use the arguments of Klaus and instead we needed to employ the Implicit Function Theorem in Banach spaces and \cite[Lemma 3.3]{GK:1969} to prove the existence of a unique positive and a unique negative eigenvalues of $\birman(z)$ which blows up if $z\in\R\setminus [\emin,\emax]$ approaches to $[\emin,\emax].$ In view of the Birman-Schwinger principle (Lemma \ref{lem:b_sh_principle}) this allows to establish the uniqueness of the eigenvalue of $H_\mu$ provided $\mu$ is small enough (Corollary \ref{cor:existence_for small_mu}).

Naturally, to get the further terms of the asymptotics of $E(\mu)$ and $e(\mu)$ one needs a further condition on the regularity of $\dispersion$ and decay of $\hat v.$ In the case with analytic $\dispersion$ and exponentially decaying  $\hat v$ one can even obtain convergent expansions as in the continuous setting \cite{KS:1980}; such an expansion for $E(\mu)$ has been obtained, for instance, in \cite{LKL:2012.tmf} in the discrete Laplacian case with zero-range non-positive perturbation.

The present paper  is organized as follows. In Section \ref{sec:birman_sh_princip} the main technical tool -- the Birman-Schwinger operator is introduced  and some of its properties are studied. The main results are proven in Section \ref{sec:proof_of_main_results}. Finally, in Appendix we obtain an asymptotics of a parametrical integral which is frequently used throughout the paper.

\subsection*{Acknowledgments}
Sh. Kholmatov acknowledges support from the Austrian Science Fund
(FWF) project M~2571-N32. S.N. Lakaev acknowledges support from the  the Foundation for Basic Research of the Republic of
Uzbekistan (Grant No.OT-F4-66).

\section{Birman-Schwinger operator and some of its properties}\label{sec:birman_sh_princip}

Given $z\in \C\setminus [\emin,\emax],$ let
$$
\birman(z):L^2(\T^d)\to L^2(\T^d),\qquad \birman(z):= {\bf s}\,\sqrt{|\potential|}(z-\hamiltonian_0)^{-1}\sqrt{|\potential|}
$$
be the Birman-Schwinger operator  associated to $\hamiltonian_\mu,$ where ${\bf s}$ is the sign of $\potential.$ By the uniqueness of the polar decomposition,
$$
|\birman(z)| =
\begin{cases}
\sqrt{|\potential| } (z - \hamiltonian_0)^{-1}\sqrt{|\potential|} & \text{for $z>\emax,$} \\
\sqrt{|\potential| } (\hamiltonian_0 - z)^{-1}\sqrt{|\potential|} & \text{for $z<\emin.$}
\end{cases}
$$
Since $\alpha_x(p):=e^{i x\cdot p},$ $x\in\Z^d,$ is the eigenvector of $\potential$  associated to $\hat v(x),$
\begin{equation}\label{trace_modul_birman}
\Tr (|\birman(z)|) = \sum\limits_{x\in\Z^d} (|\birman(z)|\alpha_x,\alpha_x) =
|a(z)|\sum\limits_{x\in\Z^d} |\hat v(x)|,\quad z\in\R\setminus[\emin,\emax],
\end{equation}
and
\begin{equation*}
\Tr (\birman(z)) =  \sum\limits_{x\in\Z^d} (\birman(z)\alpha_x,\alpha_x) =   
a(z)\sum\limits_{x\in\Z^d}  \hat v(x),\quad z\in\C\setminus[\emin,\emax],
\end{equation*}
where
$$
a(z):=\int_{\T^d} \frac{\d p}{z - \dispersion(p)}.
$$
For shortness write
\begin{equation}\label{v12_modulv12}
v^{1/2} := \cF(\sign(\hat v) \sqrt{|\hat v|})\qquad \text{and}\qquad |v|^{1/2}: = \cF(\sqrt{|\hat v|}).
\end{equation}
Then $\birman(z)$ is the integral operator with the kernel
$$
B(p,q;z) := \int_{\T^d} \frac{v^{1/2}(p-t) |v|^{1/2}(t-q)\,\d t}{z  - \dispersion(t)},\qquad p,q\in\T^d.
$$
Note that for any $z\in \C\setminus [\emin,\emax],$
$$
\|B(\cdot,\cdot;z)\|_{L^\infty((\T^d)^2)  } \le \frac{ 1 }{{\rm dist}(z,[\emin,\emax])}\,\sum\limits_{x\in\Z^d}  |\hat v(x)|.
$$
Therefore, $\birman(\cdot)$ is Hilbert-Schmidt.

\begin{lemma}[\textbf{Birman-Schwinger principle}]\label{lem:b_sh_principle}
For any  $\mu\ge0$ and $z\in\R\setminus[\emin,\emax],$
\begin{equation}\label{birman_sh_princip}
\dim \Ker (\hamiltonian_\mu - z) = \dim\Ker( 1 - \mu\birman(z) ).
\end{equation}
Moreover,
\begin{equation}\label{real_spectrum}
\sigma(\birman(z))\subset\R,\qquad z\in\R\setminus[\emin,\emax].
\end{equation}
\end{lemma}

\begin{proof}
The equality \eqref{birman_sh_princip} is well-known and can be done following the proof of for instance \cite[Lemma 2.1]{BPL:2018}. To prove \eqref{real_spectrum}, we choose any $\lambda\in\sigma(\birman(z)).$ Since $\birman(z)$ is compact, either $\lambda=0$ or $\lambda\in\C\setminus\{0\}$ is an eigenvalue of finite multiplicity. Let $f_\lambda\in L^2(\T^d)$ be a normalized eigenfunction. Then
$(|\birman(z)|f_\lambda,f_\lambda) = \lambda ({\bf s}f_\lambda,f_\lambda).$ Note that if $({\bf s} f_\lambda,f_\lambda)=0,$ then
$$
(|\birman(z)|f_\lambda,f_\lambda) = \|\sqrt{|\birman(z)|} \,f_\lambda\|^2 =0,
$$
and therefore,
$\lambda f_\lambda = {\bf s}|\birman(z)|\, f_\lambda=0,$ i.e., $f_\lambda =0.$
Hence, $({\bf s}f_\lambda,f_\lambda)\ne0.$ Since both ${\bf s}$ and $|\birman(z)|$ are self-adjoint, it follows that $\lambda=\frac{(|\birman(z)|f_\lambda,f_\lambda)}{({\bf s}f_\lambda,f_\lambda)}\in\R.$
\end{proof}

Further we study $\birman(z)$ only  for $z>\emax;$ all results straightforwardly apply to the case $z<\emin$ considering $-\hamiltonian_0 + \mu \potential.$

We start by studying of the largest eigenvalue of $\birman(z).$

\begin{lemma}\label{lem:properties_of_lambda}
The function
\begin{equation}\label{lambdaz}
z\in(\emax,+\infty)\mapsto \lambda(z) := \sup\sigma (\birman(z))
\end{equation}
is continuous and non-increasing. Moreover, $\lambda(\cdot)$ is strictly decreasing in $\{\lambda>0\}$ and
\begin{equation}\label{birman_estimate}
0\le \lambda(z) \le \|\birman(z)\|.
\end{equation}
Finally $\lambda(z_0)=0$ for some $z_0>\emax$ if and only if $\hat v\le 0.$
\end{lemma}

\begin{proof}
Since  $0\in\sigma(\birman(z))$ and the spectral radius of $\birman(z)$ does not exceed $\|\birman(z)\| = \| |\birman(z)| \|,$  \eqref{birman_estimate} follows.
Let us show that $\lambda(\cdot)$ is continuous. Fix any $z_o>\emax.$ If there exists $z_k\to z_o$ such that
$
c:=\lim\limits_{k\to\infty}\lambda(z_k) <  \lambda(z_o),
$
then $\lambda(z_o)>0$ so that it is an isolated eigenvalue of $\birman(z_o).$ Then for the spectral projection
$$
P(z_k) := \frac{1}{2\pi i} \int_{|\xi - \lambda(z_o)| = \delta_0} \big(\birman(z_k) - \xi\big)^{-1}d\xi,\qquad \delta_0:= \frac{\lambda(z_o) -c}{8}>0,
$$
is the spectral projection associated to $\lambda(z_k).$ By the definition of $c,$ $\lambda(z_k)\notin \{\xi\in\C:\,|\xi - \lambda(z_o)|<\delta_0\}$ for all large $k$ so that by the definition of $\lambda(z_k),$ $P(z_k)=0.$ However, by the norm continuity of $\birman(\cdot),$
$$
0\ne \frac{1}{2\pi i} \int_{|\xi - c|=\delta_0} \big(\birman(z_o) - \xi\big)^{-1}d\xi =\text{s-}\lim\limits_{k\to\infty} P(z_k) =0,
$$
a contradiction, where s-$\lim$ is the strong limit.  Analogous contradiction is obtained assuming the existence of $z_k\to z_o$ such that
$\lim\limits_{k\to\infty}\lambda(z_k) > \lambda(z_o).
$
Thus, $\lambda(\cdot)$ is continuous.

Now we prove that $\lambda(\cdot)$ is non-increasing. It suffices to prove that $\lambda(\cdot)$ strictly decreases in $\{\lambda>0\}.$ Fix any $z_o\in\{\lambda>0\}$ and let $\lambda(z_o)$ be an eigenvalue of $\birman(z_o)$ of multiplicity $n_o.$ Since $\birman(\cdot)$ is analytic and compact, by perturbation theory (see e.g. \cite[Chap. II, Par. 4]{Kato:1995}), there exists $\epsilon>0$ and $n_o$ differentiable  functions $\theta_1,\ldots,\theta_{n_o}:(z_o - \epsilon,\epsilon)\to(0,+\infty)$ with $\theta_i(z_o)=\lambda(z_o)$ for $i=1,\ldots,n_o$  and $\theta_i(z)$ is an eigenvalue of $\birman(z)$ with associated differentiable eigenvectors $\phi_i(z)\in L^2(\T^d).$ Then for any $i=1,\ldots,n_o$ and $z\in(z_o-\epsilon,z_o+\epsilon)$
from the eqaulity
\begin{equation}\label{and_believe_me}
\birman(z)\phi_i(z) =\theta_i(z)\phi_i(z),
\end{equation}
we get
$$
({\bf s}\phi_i(z),\phi_i(z)) = \frac{(|\birman(z_o)|\phi_i(z),\phi_i(z))}{\theta_i(z)} >0.
$$
Moreover, differentiating \eqref{and_believe_me} and scalar multiplying by ${\bf s}\phi_i(z)$ we obtain
$$
\begin{aligned}
- (\sqrt{|\potential|}(z-\hamiltonian_0)^{-2}\sqrt{|\potential|}\phi_i(z),\phi_i(z)) + (\phi_i'(z),[|\birman(z)| - \theta_i(z){\bf s}]\phi_i(z))\\
=  \theta_i'(z)({\bf s}\phi_i(z),\phi_i(z)).
\end{aligned}
$$
By \eqref{and_believe_me}, $[|\birman(z)| -  \theta_i(z){\bf s}]\phi(z)=0,$ hence,
\begin{equation*}
 \theta_i'(z) = - \frac{\|(z-\hamiltonian_0)^{-1}\sqrt{|\potential|}\phi_i(z)\|^2}{({\bf s}\phi_i(z),\phi_i(z))}.
\end{equation*}
Since $ \theta_i(z)\phi(z)\ne 0,$ this implies $ \theta_i'(z)<0$ in $(z_o-\epsilon,z_o+\epsilon).$  Hence, each $\theta_i$ is strictly decreasing. Since $\lambda(z)=\max\limits_{1\le i\le n_o} \theta_i(z_o),$ $\lambda(\cdot)$ also strictly decreases in $(z_o-\epsilon,z_o+\epsilon).$

Clearly, if $\hat v\le 0,$ then $\birman(z)\le 0$ for any $z>\emax.$ Thus, $\lambda\equiv0.$ Let us show that $\potential^+\ne 0,$ then $\lambda(z)>0$ for any $z>\emax.$ Indeed, if $\lambda(z_o)=0$ for some $z_o>\emax,$ then by monotonicity, $\lambda\equiv0$ in $(z_o,+\infty).$ Thus, for any $z>z_o$ and $\mu>0$
$$
\Ker (\hamiltonian_\mu - z)=\Ker(1 - \mu\birman(z)) =\{0\}.
$$
However, since $\potential^+\ne0,$ there exists a normalized $f_0\in L^2(\T^d)$ such that $(\potential f_0,f_0)>0.$ Thus, if we choose $\mu>\frac{z_o +1-\emin}{(\potential f_0,f_0)},$ then
$$
(\hamiltonian_\mu f_0,f_0) \ge z_o+1,
$$
i.e., by the self-adjointness of $\hamiltonian_\mu$ and \eqref{essential_spectrum}, $\Ker (\hamiltonian_\mu - z_\mu) \ne\{0\}$ for some $z_\mu>z_o+1,$ a contradiction.
\end{proof}

Further, without loss of generality we assume that the set $\{\hat v>0\}$ is non-empty so that by Lemma \ref{lem:properties_of_lambda}, $\lambda(z)>0$ for any $z>\emax.$ The following lemma shows that as $z\to\emax,$ $\birman(z)$ can be represented as a small perturbation of a rank-one operator.

\begin{lemma}\label{lem:residual_b_sh}
Assume Hypothesis \ref{hyp:maina}. Let
\begin{equation*}
Q f(p) = v^{1/2}(p-  \pmax   )\,\int_{\T^d} |v|^{1/2}(  \pmax   -q)\,f(q) \,\d q
\end{equation*}
be the rank-one projection and
\begin{equation}\label{qora_kozli_sanamjon}
Q_1(z):= \birman(z) - a(z)\,Q.
\end{equation}
%
\begin{itemize}
\item[(a)] If $d=1,$ then there exists $C_1:=C_1(\dispersion,\hat v,\gamma)>0$ such that
\begin{equation}\label{limit_q1_d1}
\|Q_1(z)\|  \le  C_1\,a(z)^{\frac{1-\gamma}{2}}
\end{equation}
for any $z\in(\emax,\emax +1).$

\item[(b)] If $d=2,$ then there exists $C_2:=C_2(\dispersion,\hat v,\gamma)>0$ such that
\begin{equation}\label{limit_q1_d2}
\|Q_1(z)\|  \le C_2
\end{equation}
for any $z>\emax$  and there exists the operator-norm limit
\begin{equation}\label{norm_converge_q1}
Q_1(\emax):=\lim\limits_{z\to+\infty} Q_1(z).
\end{equation}
Moreover, there exists $C_3:=C_3(\dispersion,\hat v,\gamma)>0$ such that
\begin{equation}\label{norm_converge_q10}
\|Q_1(z) - Q_1(\emax)\| \le C_3 (z- \emax)^{\gamma/2} \end{equation}
for any $z\in(\emax,\emax+1).$
\end{itemize}

\end{lemma}


\begin{proof}
By Hypothesis \ref{hyp:maina}  there exists a ball $B_{r_0}(\pmax)$  such that
\begin{equation}\label{unique_minimum_estimate}
c_1 |p-\pmax|^2 \le \emax - \dispersion(p)  \le c_2|p-\pmax|^2,\qquad p\in B_{r_0}(\pmax),
\end{equation}
and
\begin{equation}\label{epsilon_deee}
\emax - \max\limits_{p\in\T^d\setminus B_{r_0}(\pmax)}\, \dispersion(p)\ge c_3,
\end{equation}
where $c_1,c_2,c_3>0$ and $r_0\in(0,1)$ are constants depending only on $\dispersion$.
We rewrite $Q_1(z)$ as
$$
Q_1(z) = Q_{11}(z) + Q_{12}(z),
$$
where
\begin{align*}
Q_{11}(z) f(p) :=&  \int_{\T^d} \int_{\T^d\setminus B_{r_0}(\pmax)}\frac{v^{1/2}(p-t) |v|^{1/2}(t-q)\,\d t}{z - \dispersion(t)}\,f(q)\,\d q \\
-& \int_{\T^d\setminus B_{r_0}(\pmax) } \frac{\,\d t}{z-\dispersion(t)} \,\,\int_{\T^d} v^{1/2}(p-  \pmax   ) |v|^{1/2}(  \pmax   -q)\,f(q)\,\d q
\end{align*}
and
\begin{align*}
&Q_{12}(z)f(p)\\
:=&\int_{\T^d} \int_{B_{r_0}(\pmax)}\frac{\big[v^{1/2}(p-t) |v|^{1/2}(t-q) - v^{1/2}(p-  \pmax   ) |v|^{1/2}(  \pmax   -q)\big]\,\d t}{z - \dispersion(t)}\,f(q)\,\d q.
\end{align*}
By the defintions of $v^{1/2}$ and $|v|^{1/2},$ and \eqref{epsilon_deee} \begin{equation}\label{est_Q11}
\sup\limits_{z\le \emin} \, \|Q_{11}(z)\| \le \frac{2}{c_3}\,\sum\limits_{x\in\Z^d} |\hat v(x)|.
\end{equation}
We rewrite $Q_{12}(z)$ as
$$
Q_{12}(z): = Q_{12}^1(z) +Q_{12}^2(z),
$$
where
$$
Q_{12}^1(z) f(p) := \int_{\T^d} \int_{B_{r_0}(\pmax)} \frac{\big[v^{1/2}(p-t) - v^{1/2}(p-  \pmax   )\big]|v|^{1/2}(t-q)\,\,\d t}{z - \dispersion(t)}\,f(q)\,\d q
$$
and
$$
Q_{12}^2(z)f(p) := \int_{\T^d} \int_{B_{r_0}(\pmax)} \frac{\big[|v|^{1/2}(t-q) - |v|^{1/2}(  \pmax   -q)\big]v^{1/2}(p-  \pmax   ) \,\,\d t}{z - \dispersion(t)}\,f(q)\,\d q.
$$
By \eqref{v12_modulv12} and the Fubini Theorem,
$$
\|Q_{12}^1(z)f\|_{L^2}^2 =
\sum\limits_{x\in\Z^d} |\hat v(x)| \, \left| \int_{\T^d} \left[\int_{B_{r_0}(\pmax)} \frac{1- e^{-ix\cdot (t-  \pmax   )}}{z -\dispersion(t)} \,|v|^{1/2}(t-q)\,\d t\right] f(q)\,\d q\right|^2.
$$
Then using  the H\"older inequality we estimate
\begin{align}
& \left| \int_{\T^d} \left[\int_{B_{r_0}(\pmax)} \frac{1- e^{-ix\cdot (t-  \pmax   )}}{z -\dispersion(t)} \,|v|^{1/2}(t-q)\,\d t\right] f(q)\,\d q\right|^2 \nonumber\\
\le
& \left| \int_{B_{r_0}(\pmax)} \frac{|1- e^{-ix\cdot (t-  \pmax   )}| }{z -\dispersion(t)} \left( \int_{\T^d} \big||v|^{1/2}(t-q)\big|^2\,\d q\right)^{1/2}\,\d t \left(\int_{\T^d} |f(q)|^2\,\d q\right)^{1/2}\,\d t\right|^2\nonumber\\
= & \|f\|^2 \sum\limits_{y\in\Z^d} |\hat v(y)| \,\left[ \int_{B_{r_0}(\pmax)} \frac{2 |\sin(x\cdot (t-  \pmax   )/2)| \,\d t}{z - \dispersion(t)} \right]^2,\label{shashashasha}
\end{align}
where we used
\begin{equation}\label{norm_vhalf}
\int_{\T^d}\big|v^{1/2}(t-  q)\big|^2\,\d q =\int_{\T^d}\big||v|^{1/2}(t-  q)\big|^2\,\d q=\sum\limits_{y\in\Z^d}|\hat v(y)| = \|v\|_{\ell^1}
\end{equation}
for any $t\in\T^d.$
Note that $\frac{2-d+\gamma}{2}\in(0,1),$ where $\gamma\in(0,1)$ is given in Hypothesis \ref{hyp:maina},  thus,
\begin{equation}\label{sinus_estimate}
2\big|\sin[x\cdot (t-  \pmax   )/2]\big|\le  2^{\frac{d-\gamma}{2}} |x|^{\frac{2-d+\gamma}{2}}\,|t-p^0|^{\frac{2-d+\gamma}{2}}.
\end{equation}
Hence, by \eqref{unique_minimum_estimate}
\begin{align}\label{ajskalska_d1}
\int_{B_{r_0}(\pmax)} \frac{2 \big|\sin(x\cdot (t-  \pmax   )/2)|\,\d t }{z - \dispersion(t)} \le  &
2^{\frac{d-\gamma}{2}} |x|^{\frac{2-d+\gamma}{2}}\,\int_{B_{r_0}(\pmax)} \frac{|t-p^0|^{\frac{2-d+\gamma}{2}} \,\d t}{c_1|t-\pmax|^2 + z-\emax}\nonumber \\
= & \frac{2^{\frac{d-\gamma}{2}}d\omega_d}{c_1}\,  |x|^{\frac{2-d+\gamma}{2}}\, T_{\frac{d+\gamma}{2}}\big((\tfrac{z-\emax}{c_2})^{1/2}\big),
\end{align}
where in the equality we passed to polar coordinates, $w_d$ is the volume of the unit ball in $\R^d$, i.e., $\omega_1:=2$ and $w_2:=\pi,$ and $T_\alpha$ is given in \eqref{def:t_d}.  Now if $d=1,$ then $\frac{d+\gamma}{2}\in(0,1)$ and thus, by Lemma \ref{lem:asymp_some_integral}
$$
T_{\frac{d+\gamma}{2}}\big((\tfrac{z-\emax}{c_2})^{1/2}\big) \le c_{1,\gamma}\, (z - \emax)^{\frac{\gamma-1}{4}},\qquad z>\emax.
$$
for some $c_{1,\gamma}>0.$ If $d=2,$ then $\frac{d+\gamma}{2}\in(1,2)$ and hence
$$
T_{\frac{d+\gamma}{2}}\big((\tfrac{z-\emax}{c_2})^{1/2}\big) \le \frac{2r_0^{\gamma/2}}{\gamma}, \qquad z>\emax.
$$
Hence,
\begin{equation}\label{est_Q1p1}
\|Q_{12}^1(z)\| \le
\begin{cases}
A_1\, (z-\emax)^{\frac{\gamma - 1}{4}} & \text{if $d=1,$}\\
A_2  & \text{if $d=2,$}
\end{cases}
\end{equation}
where
\begin{equation}\label{ad_coefficient}
A_d:=
\begin{cases}
\frac{c_{1,\gamma}}{c_1}\,\Big(2^{3- \gamma }\sum\limits_{y\in\Z^d} |\hat v(y)| \sum\limits_{x\in\Z^d}|x|^{2-d+ \gamma } |\hat v(x)| \Big)^{1/2} & \text{if $d=1,$}\\[1mm]
\frac{4\pi r_0^{\gamma/2}}{c_1\gamma}\,\Big(2^{2- \gamma }\sum\limits_{y\in\Z^d} |\hat v(y)| \sum\limits_{x\in\Z^d}|x|^{2-d+ \gamma } |\hat v(x)| \Big)^{1/2} & \text{if $d=2.$}
\end{cases}
\end{equation}

Now we estimate $\|Q_{12}^2(z)f\| .$ By \eqref{v12_modulv12} and the Fubini Theorem,
$$
\begin{aligned}
& |Q_{12}^2(z)f(p)|\\
\le  & |v^{1/2}(p-  \pmax   )|\sum\limits_{x\in\Z^d} \sqrt{ |\hat v(x)|}\,\left|\int_{B_{r_0}(\pmax)} \frac{1- e^{ix\cdot (t-  \pmax   )}}{z - \dispersion(t) } \,e^{-ix\cdot t}\,\,\d t \,\int_{\T^d} e^{-ix\cdot q} f(q)\,\d q\right| \\
\le & |v^{1/2}(p-  \pmax   )|\sum\limits_{x\in\Z^d} \sqrt{|\hat v(x)|}\,\, |\hat f(x)|\,\int_{B_{r_0}(\pmax)} \frac{|1- e^{ix\cdot (t-  \pmax   )}|\,\d t}{z - \dispersion(t)}\\
 = & |v^{1/2}(p-  \pmax   )|\sum\limits_{x\in\Z^d} \sqrt{|\hat v(x)|}\,\,|\hat f(x)|\,\int_{B_{r_0}(\pmax)} \frac{2\big|\sin[x\cdot (t-  \pmax   )/2]\big|\,\d t}{z - \dispersion(t)},
\end{aligned}
$$
where $\hat f = \cF^*f.$  By \eqref{ajskalska_d1} and the H\"older inequality
\begin{align*}
&\sum\limits_{x\in\Z^d} \sqrt{|\hat v(x)|}\,\,|\hat f(x)|\,\int_{B_{r_0}(\pmax)} \frac{2|\sin[x\cdot (t-  \pmax   )/2]|\,\d t}{z - \dispersion(t)}\\
\le& \sum\limits_{x\in\Z^d}
\Big(2^{d- \gamma }|x|^{2-d+ \gamma } |\hat v(x)|\Big)^{1/2} |\hat f(x)|  \int_{B_{r_0}(\pmax)} \frac{|t-  \pmax   |^{\frac{2-d+ \gamma }{2}}\,\d t }{c_1|t-\pmax|^2 +z - \emax}
\\
\le &
\frac{1}{c_1}\,\Big(2^{d- \gamma } \sum\limits_{x\in\Z^d} |x|^{2-d+ \gamma } |\hat v(x)|\Big)^{1/2} \|\hat f\|_{\ell^2} T_{\frac{2-d+\gamma}{2}} \big((\tfrac{z-\emax}{c_1})^{1/2}\big).
\end{align*}
Thus, using $\|\hat f\|_{\ell^2(\Z^d)} = \|f\|_{L^2(\T^d)}$ and \eqref{norm_vhalf} we get
\begin{equation}\label{est_Q1p2}
\|Q_{12}^2(z)\| \le
\begin{cases}
A_1\, (z-\emax)^{\frac{\gamma - 1}{4}} & \text{if $d=1,$}\\
A_2  & \text{if $d=2,$}
\end{cases}
\end{equation}
where $A_d$ is given in \eqref{ad_coefficient}.  Since,
\begin{equation}\label{rep_q1_as_sum}
Q_1(z) = Q_{11}(z) +  Q_{12}^1(z) + Q_{12}^2(z),
\end{equation}
from \eqref{est_Q11}, \eqref{est_Q1p1} and \eqref{est_Q1p2} it follows that
\begin{equation} \label{estimate_q1_norm}
\|Q_1(z)\| \le \frac{2}{c_3}\,\sum\limits_{x\in\Z^d} |\hat v(x)| +
\begin{cases}
2A_1\, (z-\emax)^{\frac{\gamma - 1}{4}} & \text{if $d=1,$}\\
2A_2  & \text{if $d=2.$}
\end{cases}
\end{equation}

(a)  Let $d=1.$ Let us estimate $a(z)$ from below. By \eqref{unique_minimum_estimate}
\begin{align*}
a(z) \ge & \int_{B_{r_0}(\pmax)} \frac{\,\d q}{c_2| p -\pmax|^2 + z-\emax} = \frac{d\omega_d}{c_2}\,T_{d-1}\big( (\tfrac{z-\emax}{c_2})^{1/2} \big)\nonumber \\[2mm]
= &
\frac{\pi}{c_2^{1/2}(z-\emax)^{1/2}}\,\Big[1 - \frac{2}{\pi}\, \arctan \frac{(z-\emax)^{1/2}}{c_2^{1/2}r_0}\Big].
\end{align*}
This and \eqref{estimate_q1_norm} implies \eqref{limit_q1_d1}.

(b) Let $d=2.$ The estimate \eqref{limit_q1_d2} directly follows from \eqref{estimate_q1_norm}. Now we prove \eqref{norm_converge_q1}-\eqref{norm_converge_q10}. By the definition of
$Q_1(z),$ the norm limit
$$
Q_{11}(\emax):=\lim\limits_{z\to\emax} Q_{11}(z)
$$
exists and
\begin{align*}
Q_{11}(\emax) f(p):= &\int_{\T^d} \int_{\T^d\setminus B_{r_0}(\pmax)}\frac{v^{1/2}(p-t) |v|^{1/2}(t-q)\,\d t}{\emax - \dispersion(t)}\,f(q)\,\d q \\
-& \int_{\T^d\setminus B_{r_0}(\pmax) } \frac{\,\d t}{\emax-\dispersion(t)} \,\,\int_{\T^d} v^{1/2}(p-  \pmax   ) |v|^{1/2}(  \pmax   -q)\,f(q)\,\d q.
\end{align*}
It is obvious that
\begin{equation}\label{norm_limit_q11}
\|Q_{11}(z) - Q_{11}(\emax)\| \le \tilde c_1\,(z-\emax),\qquad  z>\emax,
\end{equation}
for some $\tilde c_1>0$ independent of $z.$
Furthermore, repeating the same proof of the uniform boundedness of $Q_{12}^1(z)$ and $Q_{12}^2(z)$ one can show the boundedness of operators
$$
Q_{12}^1(\emax) f(p) := \int_{\T^d} \int_{B_{r_0}(\pmax)} \frac{\big[v^{1/2}(p-t) - v^{1/2}(p-  \pmax   )\big]|v|^{1/2}(t-q)\,\,\d t}{\emax - \dispersion(t)}\,f(q)\,\d q
$$
and
$$
Q_{12}^2(\emax)f(p) := \int_{\T^d} \int_{B_{r_0}(\pmax)} \frac{\big[|v|^{1/2}(t-q) - |v|^{1/2}(  \pmax   -q)\big]v^{1/2}(p-  \pmax   ) \,\,\d t}{\emax - \dispersion(t)}\,f(q)\,\d q.
$$
We claim that for any $z>\emax$
\begin{align}
\|Q_{12}^1(z) - Q_{12}^1(\emax) \|\le \tilde c_2\,(z -\emax)^{\gamma/2}, \label{norm_limit_q111}\\
\|Q_{12}^2(z) - Q_{12}^2(\emax) \|\le \tilde c_2\,(z -\emax)^{\gamma/2}  \label{norm_limit_q112}
\end{align}
for some $\tilde c_2>0$ independent of $z.$
We prove only \eqref{norm_limit_q111}, the proof of \eqref{norm_limit_q112} being similar. For any $f\in L^2(\T^d)$ let us estimate the $L^2$-norm of
\begin{align*}
&(Q_{12}^1(z) - Q_{12}^1(\emax))f(p)\\
=&-(z - \emax)\, \int_{\T^d}\int_{B_{r_0}(\pmax)} \frac{\big[v^{1/2}(p-t) - v^{1/2}(p-  \pmax   )\big]|v|^{1/2}(t-q)\,\,\d t}{(z-\dispersion(t))(\emax - \dispersion(t))}\,f(q)\,\d q.
\end{align*}
As in \eqref{shashashasha}-\eqref{sinus_estimate}
\begin{align*}
&\|(Q_{12}^1(z) - Q_{12}^1(\emax))f\|^2
\le 2^{2-\gamma}(z-\emax) \|\hat v\|_{\ell^1}\|f\|^2\times\\
&\times \sum\limits_{x\in\Z^2} |x|^{\frac{\gamma}{2}}|\hat v(x)|\, \left|\int_{B_{r_0}(\pmax)} \frac{|t-\pmax|^{\frac{\gamma}{2}} \,\d t}{c_1|t-\pmax|^2\,(c_1|t-\pmax|^2 + z-\emax)}\right|^2.
\end{align*}
Then passing to the polar coordinates we obtain
$$
\|(Q_{12}^1(z) - Q_{12}^1(\emax))f\|
\le \tilde c_0\,(z-\emax)\|f\| \, \int_0^{r_0} \frac{r^{\frac{\gamma}{2}-1}\, \d r}{c_1r^2 + z-\emax},
$$
where $\tilde c_0:=2\pi \Big(2^{2-\gamma}\|\hat v\|_{\ell^1} \sum\limits_{x\in\Z^2} |x|^{\frac{\gamma}{2}}|\hat v(x)|\Big)^{1/2}.$
Now using the change of variables $r=(z-\emax)^{1/2} t$ we get
$$
\|(Q_{12}^1(z) - Q_{12}^1(\emax))\|
\le \tilde c_0\,(z-\emax)^{\gamma/2}\,\int_0^{+\infty} \frac{t^{\frac{\gamma}{2}-1} \,\d t}{c_1t^2 +1}.
$$
Since $\int_0^{+\infty} \frac{t^{\frac{\gamma}{2}-1} \,\d t}{c_1t^2 +1}<\infty,$ \eqref{norm_limit_q111} follows.

Now we set
$$
Q_1(\emax):=Q_{11}(\emax) + Q_{11}^1(\emax)  + Q_{11}^2(\emax).
$$
Then from \eqref{rep_q1_as_sum}, \eqref{norm_limit_q11}, \eqref{norm_limit_q111} and \eqref{norm_limit_q112} we conclude \eqref{norm_converge_q1} and \eqref{norm_converge_q10}.
\end{proof}

%

Next we study the case of sign-definite potentials.

\begin{lemma}\label{lem:asymp_lambda_sign_definite_case}
Assume Hypothesis \ref{hyp:maina} and that $\hat v\ge0.$  Then there exists $C_4:=C_4(\dispersion,\gamma)>0$ such that
\begin{equation}\label{boje_moi}
\big|\Tr(\birman(z)) - \|\birman(z)\| \big|\le C_4 \sum\limits_{x\in\Z^d} |x|^{2-d+ \gamma }|\hat v(x)|<\infty.
\end{equation}
\end{lemma}

\begin{proof}
Let $\nu(z):=\|\birman(z)\|;$ by self-adjointness and  \eqref{lambdaz} $\lambda(z)=\nu(z)$ for any $z>\emax.$
Since $\hat v\ge0,$ one has $\birman(z)\ge0$ and $v^{1/2} =|v|^{1/2}.$ By Lemma \ref{lem:residual_b_sh}, the self-adjoint operator $\frac{\birman(z)}{a(z)} = Q + \frac{Q_1(z) }{a(z)}$ is a small pertubation of the self-adjoint rank-one projector
$$
Q:=\phi_0(\cdot,\phi_0),
$$
where $\phi_0(p):=v^{1/2}(p -\pmax),$ therefore, by the standard perturbation theory, the eigenvector $f_z$ associated to $\nu(z)$ can be written as $f_z = \phi_0+ \phi_z,$ where $\phi_z$ is orthogonal to $\phi_0.$ Then using
$$
(\phi_0,\phi_0) =   \sum\limits_{x\in\Z^d} \hat v(x)
$$
one has
\begin{equation}\label{asymp_eigenvalue}
\nu(z)\phi_0  + \nu(z)\phi_z =  a(z)\phi_0 \sum\limits_{x\in \Z^d} \hat v(x) + Q_1(z)\phi_0  + Q_1(z)\phi_z.
\end{equation}
Multiplying \eqref{asymp_eigenvalue} by $\phi_z$  we get
$$
\nu(z)\|\phi_z\|^2 = (Q_1(z)\phi_0 ,\phi_z) +(Q_1(z)\phi_z,\phi_z).
$$
Since $(\phi_z,\phi_0 )=0,$ one has $(\birman(z)\phi_z,\phi_z)=(Q_1(z)\phi_z,\phi_z)\ge0.$ Thus,
$$
\|\birman(z)\|\,\,\|\phi_z\|^2 =  (Q_1(z)\phi_0 ,\phi_z) + (\birman(z)\phi_z,\phi_z),
$$
and in particular, $(Q_1(z)\phi_0 ,\phi_z)\le0.$ Multiplying \eqref{asymp_eigenvalue} by $\phi_0$ we get
\begin{equation}\label{nihoyat}
\nu(z)\sum\limits_{x\in\Z^d} \hat v(x)  = a(z)\Big( \sum\limits_{x\in \Z^d}  \hat v(x) \Big)^2 + (Q_1(z)\phi_0,\phi_0)  + (Q_1(z)\phi_z,\phi_0).
\end{equation}
Since all nonzero eigenvalues of $\birman(z)$ are positive and $ (Q_1(z)\phi_z,\phi_0)\ge0,$  by \eqref{trace_modul_birman} and \eqref{nihoyat},
\begin{align}\label{eqrqeqeqe}
\Big| \Tr(\birman(z)) - \nu(z)\Big|\,\,\sum\limits_{x\in\Z^d} |\hat v(x)| = \Big[a(z) \sum\limits_{x\in\Z^d} \hat v(x) -\nu(z) \Big]\,\,\sum\limits_{x\in\Z^d} \hat v(x)  \nonumber \\
=-(Q_1(z)\phi_0,\phi_0)  - (Q_1(z)\phi_z,\phi_0) \le -(Q_1(z)\phi_0,\phi_0).
\end{align}
Note that
\begin{align*}
- (Q_1(z)\phi_0,\phi_0)=&  \int_{(\T^d)^3}\frac{v^{1/2}(p-\pmax)v^{1/2}(\pmax-q) - v^{1/2}(p-t)v^{1/2}(t-q) }{z - \dispersion(t)}\\
& \times v^{1/2}(q-\pmax)v^{1/2}(\pmax-p)\,\,\d t\,\d q\,\d p\\
=& \int_{\T^d} \frac{v(0)^2 - |v(t - \pmax)|^2}{z-\dispersion(t) }\,\d t,
\end{align*}
where we used
$$
\int_{\T^d} v^{1/2}(p-t)v^{1/2}(s-p)\,\d p= v(s-t),\qquad s,t\in\T^d.
$$
Since $\hat v\ge0,$
$$
|v(p)|^2= \sum\limits_{x,y\in\Z^d} \hat v(x) \hat v(y) \cos[ (x-y)\cdot p],
$$
and thus, $p=0$ is the maximum point of $|v(\cdot)|^2.$ Therefore, the map $z\in(\emax,+\infty)\mapsto -(Q_1(z)\phi_0,\phi_0)$ is decreasing
 and
\begin{equation}\label{otamning_giroti}
-(Q_1(z)\phi_0,\phi_0)= 2 \sum\limits_{x,y\in\Z^d} \hat v(x)\hat v(y) \int_{\T^d} \frac{\sin^2\big[\frac{(x-y)(t-\pmax)}{2}\big] \,\d t}{z -\dispersion(t)}.
\end{equation}
Since $\sin^2\big[\frac{(x-y)(t-\pmax)}{2}\big] \le \big|\frac{(x-y)(t-\pmax)}{2}\big|^{2-d+ \gamma },$  from \eqref{otamning_giroti} it follows that
\begin{equation*}
-(Q_1(z)\phi_0,\phi_0)\le 2^{d-1- \gamma } \sum\limits_{x,y\in\Z^d} \hat v(x)\hat v(y)|x-y|^{2-d+ \gamma } \int_{\T^d} \frac{|t-\pmax|^{2-d+ \gamma }\,\d t}{z - \dispersion(t)}.
\end{equation*}
Using
$$
|x-y|^{2-d+ \gamma } \le 2^{2-d+ \gamma } \big(|x|^{2-d+ \gamma } +|y|^{2-d+ \gamma } \big)
$$
one can readily check that
$$
\sum\limits_{x,y\in\Z^d} \hat v(x)\hat v(y)|x-y|^{2-d+ \gamma } \le
2^{3-d+ \gamma } \sum\limits_{x\in\Z^d} |\hat v(x)|\,\sum\limits_{x\in\Z^d} |x|^{2-d+ \gamma }|\hat v(x)|.
$$
Thus, from \eqref{otamning_giroti} we get
$$
-(Q_1(z)\phi_0,\phi_0)\le 4\sum\limits_{x\in\Z^d} |\hat v(x)|\,\sum\limits_{x\in\Z^d} |x|^{2-d+ \gamma }|\hat v(x)|\, \int_{\T^d} \frac{|t-\pmax|^{2-d+ \gamma }\,\d t}{z - \dispersion(t)}.
$$
Let $B_{r_0}(\pmax)$ be such that \eqref{unique_minimum_estimate} and \eqref{epsilon_deee} hold. We write
\begin{align}\label{har_muyi_uza_yuz}
\int_{\T^d} \frac{|t-\pmax|^{2-d+ \gamma }\,\d t}{z - \dispersion(t)}
= &
\int_{B_{r_0}(\pmax)} \frac{|t-\pmax|^{2-d + \gamma }\,\d t}{z - \dispersion(t)}+ \int_{\T^d\setminus B_{r_0}(\pmax)} \frac{|t-\pmax|^{2-d+ \gamma }\,\d t}{z - \dispersion(t)}\nonumber \\
=&:I_1(z)+I_2(z).
\end{align}
By \eqref{epsilon_deee}, $\sup\limits_{z<\emin} I_2(z)<\frac{16\pi^2}{c_3}.$ Using  \eqref{unique_minimum_estimate} in $I_1(z)$ and passing to polar coordinates  we get
$$
I_1(z) =\int_{B_{r_0}(\pmax)} \frac{|t-\pmax|^{2-d + \gamma }\,\d t}{c_1|t-\pmax|^2 + z -\emax} = \frac{d\omega_d}{c_1} \, T_{1+\gamma}\big((\tfrac{z-\emax}{c_1})^{1/2}\big),
$$
where $T_\alpha$ is defined in \eqref{def:t_d}. Since $1+\gamma>1,$ by Lemma \ref{lem:asymp_some_integral}
$$
I_1(z)\le \frac{d\omega_dr_0^\gamma}{c_1 \gamma } < \infty
$$
for any $z>\emax.$ Thus,  from \eqref{eqrqeqeqe} and \eqref{har_muyi_uza_yuz},
$$
\Big| \Tr(\birman(z)) - \nu(z)\Big| \le 4\,\Big(\frac{16\pi^2}{c_1} + \frac{d\omega_d r_0^\gamma }{ c_1\gamma }
\Big) \sum\limits_{x\in\Z^d} |x|^{2-d+ \gamma }|\hat v(x)|
$$
and  \eqref{boje_moi} follows.
\end{proof}

Next we study all positive eigenvalues of $\birman(z)$ and their limits as $z\to \emax.$

\begin{proposition}\label{prop:asym_eigenvalues_bsh}
Assume Hypothesis \ref{hyp:maina}  and let $\lambda(\cdot)$ be defined as \eqref{lambdaz}. Let $\lambda_0(z)= \lambda(z)\ge \lambda_1(z)\ge \ldots >0$ be  all positive eigenvalues (if any) of $\birman(z)$ counted with their multiplicities. Then there exists $C_5:=C_5(\dispersion,\hat v,\gamma)>1$ such that for any $z>\emax$
\begin{equation}\label{jilmayasan_xayol_surasan}
0\le \lambda_k(z) \le C_5,\qquad k\ge1.
\end{equation}
Moreover:
\begin{itemize}
\item[(a)] if $\sum_{x\in\Z^d} \hat v(x)<0,$ then
$$
0<\sup\limits_{z>\emax}  \lambda(z) \le C_5;
$$

\item[(b)] if $\sum_{x\in\Z^d} \hat v(x)\ge0,$ then there exists $\delta_0:=\delta_0(\dispersion,\hat v)>0$ such that:

\begin{itemize}
\item[(b1)] if $\kappa_0:=\sum_{x\in\Z^d} \hat v(x)>0,$ then
for any $z\in (\emax,\emax + \delta_0)$
\begin{equation}\label{dangal_dangal}
\frac{\lambda(z)}{a(z)} =  \kappa_0 +
\begin{cases}
\frac{g_1(z)\ln (z-\emax)}{\kappa_0 a(z)} & \text{if $d=1,$}\\[1mm]
\frac{C_6}{\kappa_0  a(z)} + \frac{(z-\emax)^{\gamma/4}\,g_2(z)}{\kappa_0 a(z)} & \text{if $d=2,$}
\end{cases}
\end{equation}
where $C_6\in\R,$ $g_1,g_2\in C^0[\emax,\emax+\delta_0];$

\item[(b2)] if $\kappa_0:=\sum_{x\in\Z^d} \hat v(x)=0,$ then
$$
\kappa_1:=\int_{\T^d} \frac{|v(p-\pmax)|^2\,\d p}{\emax - \dispersion(p)}
$$
is finite and  for any $z\in (\emax,\emax + \delta_0)$
\begin{equation}\label{issaqbaazi_he}
\frac{\lambda(z)^2}{a(z)} =
\kappa_1 +
\begin{cases}
\frac{g_3(z)\ln^2(z-\emax)}{a(z)^{\gamma/2}} & \text{if $d=1,$}\\[1mm]
\frac{C_7}{a(z)^{1/2}} + \frac{(z-\emax)^{\gamma/2}g_4(z)}{a(z)^{1/2}} & \text{if $d=2,$}
\end{cases}
\end{equation}
where $C_7\in\R,$ $g_3,g_4\in C^0[\emax,\emax+\delta_0].$
\end{itemize}
\end{itemize}
\end{proposition}


\begin{proof}

First assume that
$$
\kappa_0:=\sum\limits_{x\in\Z^d} \hat v(x)\ne0.
$$
Then  $\kappa_0$ is the unique nonzero eigenvalue of $Q$ and its  associated eigenvector is $\phi_0(p):=v^{1/2}(p-\pmax).$ Note that
\begin{equation}\label{resolvent_Q}
(Q-\xi)^{-1} = - \frac{1}{\xi} - \frac{Q}{\xi(\xi-\kappa_0)},\qquad \xi\ne0,\kappa_0.
\end{equation}
By Lemma \ref{lem:residual_b_sh}, $\frac{\birman(z)}{a(z)}$ is a small perturbation of $Q,$ and hence, by the standard perturbation theory methods for the non-selfadjoint operators, the eigenvalue $\eta(z)$ of $\birman(z)$ of maximal modulus satisfies
\begin{equation}\label{maximal_eigenvalues}
\lim\limits_{z\to\emax}\, \frac{\eta(z)}{a(z)} = \kappa_0.
\end{equation}
Let $\epsilon:=\frac{|\kappa_0|}{8}>0.$ By Lemma \ref{lem:residual_b_sh} there exists $\delta_1\in(0,1)$ such that
$$
\Big| \frac{\eta(z)}{a(z)}  - \kappa_0\Big| <\epsilon\qquad\text{and}\qquad \frac{\|Q_1(z)(Q-\xi)^{-1}\|}{a(z)}<\frac{1}{2}
$$
for any $z\in(\emax,\emax+\delta_1)$ and any $\xi\in\C$ such that $|\xi - \kappa_0|=\epsilon.$  Then for any such $\xi$ one has
$$
\begin{aligned}
\left(\frac{\birman(z)}{a(z)} - \xi \right)^{-1} = & (Q-\xi)^{-1}\left(1  + \frac{Q_1(z)}{a(z)}\,(Q-\xi)^{-1}\right)^{-1} \\
= & (Q-\xi)^{-1} + \sum\limits_{n\ge1}\frac{(-1)^n}{a(z)^n}(Q-\xi)^{-1}[Q_1(z)(Q-\xi)^{-1}]^n.
\end{aligned}
$$
Therefore, if we integrate this equality over the complex circle $|\xi - \kappa_0|=\epsilon$ and insert \eqref{resolvent_Q}, then by the Residue Theorem for analytic functions and the estimates \eqref{limit_q1_d1} and \eqref{limit_q1_d2},
\begin{align*}
P(z):= & \frac{1}{2\pi i}\,\int_{|\xi - \kappa_0|=\epsilon} \left(\frac{\birman(z)}{a(z)} - \xi \right)^{-1}d\xi \\
= & -\frac{Q}{\kappa_0} - \frac{1}{a(z)}\,\frac{Q_1(z)Q + QQ_1(z)}{\kappa_0^2} + \frac{R(z)}{a(z)},
\end{align*}
where $\|R(z)\|=O((z-\emax)^{\gamma/4})$ as $z\to\emax.$ By the definition of $\epsilon,$ $P(z)$  is the spectral projection of $\birman(z)$ associated to its eigenvalue $\frac{\eta(z)}{a(z)}.$ Hence, using
$$
Q\phi_0=\kappa_0 \phi_0
$$
and
$$
(Q_1(z)Q + QQ_1(z))\phi_0 = \kappa_0Q_1(z)\phi_0 + (Q_1(z)\phi_0,\phi_1)\phi_0,
$$
where $\phi_1(p):=|v|^{1/2}(p-\pmax),$ we deduce that the associated eigenvector $f_z$ is
$$
f_z:= - P(z)\phi_0= \phi_0 + \frac{\kappa_0Q_1(z)\phi_0 + (Q_1(z)\phi_0,\phi_1)\phi_0}{\kappa_0^2a(z)} + \frac{T(z)\phi_0}{a(z)}.
$$
Then scalar multiplying the eigenvalue equation $\frac{\eta(z)}{a(z)}\,f_z = \big(Q + \frac{Q_1(z)}{a(z)}\big)f_z$ by $\phi_1$ and using $(\phi_0,\phi_1)_{L^2} = \kappa_0$ we get
$$
 \left(\frac{\eta(z)}{a(z)} - \kappa_0\right) \kappa_0 + \frac{(Q_1(z)\phi_0,\phi_1)}{\kappa_0 a(z)}\,\left(\frac{2\eta(z)}{a(z)} - 3\kappa_0\right)= \frac{h_0(z)}{a(z)},
$$
where $h_0\in C^0(\emax,\emax+\delta_1)$ satisfies $|h_0(z)|=O((z-\emax)^{\gamma/4})$ as $z\to\emax.$
Hence,
\begin{align}\label{qolu_ya_abana}
\frac{\eta(z)}{a(z)} = \kappa_0 + \frac{(Q_1(z)\phi_0,\phi_1)}{\kappa_0  a(z)} + \frac{h_1(z)}{\kappa_0 a(z)},
\end{align}
where $h_1\in C^0(\emax,\emax+\delta_1)$ satisfies $|h_1(z)|=O((z-\emax)^{\gamma/4})$ as $z\to\emax.$

Notice that
\begin{equation}\label{tastati_a}
(Q_1(z)\phi_0,\phi_1)= \int_{\T^d} \frac{(|v(p-\pmax)|^2 - |v(0)|^2)\,\d p}{z-\dispersion(p)}.
\end{equation}
Since $v\in C^0(\T^d),$ by $v=\cF\hat v$ we have
\begin{align*}
\big||v(p-\pmax)|^2 - |v(0)|^2\big| \le &2\|v\|_{L^\infty} \,|v(p-\pmax) - v(0)|\nonumber \\
\le& 2\|v\|_{L^\infty} \sum\limits_{x\in\Z^d} |\hat v(x)|\, \big|\sin \tfrac{x\cdot (p-\pmax)}{2}\big|\nonumber \\
\le &
\begin{cases}
2\|v\|_{L^\infty} \,|p-\pmax|\,\sum\limits_{x\in\Z^d} |x|\, |\hat v(x)| &\text{if $d=1,$}\\[2mm]
2\|v\|_{L^\infty} \,|p-\pmax|^\gamma\,\sum\limits_{x\in\Z^d} |x|^\gamma\, |\hat v(x)| &\text{if $d=2.$}
\end{cases}
\end{align*}
Separating integral in \eqref{tastati_a} into integrals over $\T\setminus B_{r_0}(\pmax)$ and $B_{r_0}(\pmax)$ and using \eqref{unique_minimum_estimate} and \eqref{epsilon_deee} we obtain
$$
|(Q_1(z)\phi_0,\phi_1)| \le  \tilde c_1 + \tilde c_2 \int_{B_{r_0}(\pmax)} \frac{|p-\pmax|^\alpha \,\d p}{c_1|p-\pmax|^2 + z-\emax},
$$
where $\tilde c_1,\tilde c_2>0$ and $\alpha=1$ for $d=1$ and $\alpha=\gamma$ if $d=2.$ Now passing to polar coordinates and using Lemma \ref{lem:asymp_some_integral} for any $z\in(\emax,\emax+\delta_1)$ we get
\begin{align}\label{ishaq_inna}
|(Q_1(z)\phi_0,\phi_1)| \le  & \tilde c_1 + \frac{\tilde c_2 d\omega_d}{c_1}\, T_{d-1+\alpha}\big((\tfrac{z-\emax}{c_1})^{1/2}\big)\nonumber \\
\le &
\begin{cases}
\tilde c_1 - \tilde c_3 \ln(z-\emax) &\text{if $d=1,$}\\
\tilde c_3 &\text{if $d=2,$}
\end{cases}
\end{align}
where $T_\alpha$ is defined in \eqref{def:t_d} and $\tilde c_3>0.$ Thus,  if $d=1,$ then \eqref{qolu_ya_abana} is represented as
\begin{equation}\label{eigenvalue_asymp_d1}
\frac{\eta(z)}{a(z)} = \kappa_0 + \frac{h_2(z)}{\kappa_0 a(z)},
\end{equation}
where $h_2\in C^0(\emax,\emax+\delta_1)$ satisfies $|h_2(z)|\le - \tilde c_4 \ln (z-\emax)$ for some $\tilde c_4>0.$
If $d=2,$ then there exists the norm-limit
$$
(Q_1(\emax)\phi_0,\phi_1):=\lim\limits_{z\to\emax} (Q_1(z)\phi_0,\phi_1).
$$
Repeating the arguments of \eqref{norm_converge_q10} one can show
\begin{equation}\label{estimate_ez_e0}
\big| (Q_1(z)\phi_0,\phi_1) - (Q_1(\emax)\phi_0,\phi_1)\big| \le \tilde c_5 (z-\emax)^{\gamma/4}
\end{equation}
for any $z\in(\emax,\emax+\delta).$ Thus, \eqref{qolu_ya_abana} is rewritten as
\begin{equation}\label{eigenvalue_asymp_d2}
\frac{\eta(z)}{a(z)} = \kappa_0 + \frac{(Q_1(\emax)\phi_0,\phi_1)}{\kappa_0  a(z)} + \frac{h_3(z)}{\kappa_0 a(z)},
\end{equation}
where $h_3\in C^0(\emax,\emax+\delta_1)$ satisfies $|h_3(z)|=O((z-\emax)^{\gamma/4})$ as $z\to\emax.$ Hence, \eqref{dangal_dangal} follows.

Now we prove \eqref{jilmayasan_xayol_surasan}.
Let $\eta_0(z)=\eta(z),\eta_1(z),\ldots$ be all nonzero eigenvalues of $\birman(z)$ counted with multiplicities as $|\eta_0(z)| \ge|\eta_1(z)|\ge \ldots>0$
and let $\nu_1(z)=\big\||\birman(z)|\big\|\ge \nu_2(z) \ge \ldots>0$ be  all eigenvalues of $|\birman(z)|,$ then by \cite[Lemma 3.3]{GK:1969} we have
$$
0<|\eta_0(z)|<\nu_1(z)\quad\text{and}\quad 0< |\eta_0(z)\eta_1(z)| \le \nu_0(z)\nu_1(z).
$$
By Lemma \ref{lem:asymp_lambda_sign_definite_case},
\begin{equation*}
\nu_0(z)= a(z)\Big[\sum\limits_{x\in\Z^d} |\hat v(x)| + o(1) \Big]
\end{equation*}
and
\begin{equation*}
\nu_1(z) \le  C_4 \sum\limits_{x\in\Z^d} |x|^{2-d+ \gamma }|\hat v(x)|,
\end{equation*}
and by \eqref{eigenvalue_asymp_d1}-\eqref{eigenvalue_asymp_d2}
\begin{equation*}
\eta_0(z)= a(z)\Big[\sum\limits_{x\in\Z^d} \hat v(x)  + o(1) \Big].
\end{equation*}
Therefore, there exists $\delta_2\in(0,\delta_1)$ such that
$$
|\eta_1(z)| \le \frac{\nu_0(z) \nu_1(z)}{|\eta_0(z)|} \le \tilde C_5:=C_4 \Big|\sum\limits_{x\in\Z^d} \hat v(x)\Big|^{-1}\,  \sum\limits_{x\in\Z^d} |x|^{2-d+ \gamma }|\hat v(x)| + 1
$$
for any $z\in(\emax,\emax + \delta_2).$ Since $|\eta_1(z)| \le |\eta_0(z)|\le \big\||\birman(z)|\big\|$ for any $z>\emax$ and the map $z\mapsto \big\||\birman(z)|\big\|$ is non-increasing,
$$
\sup\limits_{z\ge \emax+\delta_2} \,\,|\eta_1(z)| \le \big\||\birman(\emax+\delta_2)|\big\|=\big\|\birman(\emax+\delta_2)\big\|.
$$
Since $|\eta_k|\le |\eta_1|$ for any $k\ge1,$ we get
\begin{equation}\label{quyoshga_qarab_oqqan_suv}
|\eta_k(z)| \le   C_5:=\max\{\tilde C_5,\|\birman(\emax+\delta_2)\|\},\qquad k\ge1.
\end{equation}
In particular, from \eqref{jilmayasan_xayol_surasan} we get \eqref{quyoshga_qarab_oqqan_suv}.

Finally, we observe that
\begin{itemize}
\item[--] if $\kappa_0=\sum\limits_{x\in\Z^d} \hat v(x)<0,$  then by \eqref{maximal_eigenvalues}, $\eta_0(z)<0,$ and hence the assertion (a) follows from \eqref{quyoshga_qarab_oqqan_suv};

\item[--] if $\kappa_0=\sum\limits_{x\in\Z^d} \hat v(x)>0,$  then $\lambda(z)=\eta_0(z)$ and the assertion (b) follows from  \eqref{maximal_eigenvalues}, \eqref{eigenvalue_asymp_d1} and \eqref{eigenvalue_asymp_d2}.
\end{itemize}

Now we consider the case $\kappa_0=\sum\limits_{x\in\Z^d} \hat v(x)=0.$ In this case since $v(0)=0,$
\begin{equation}\label{q1f_phi1}
(Q_1(z)f,\phi_1)_{L^2} = \int_{\T^d}\int_{\T^d} \frac{v(\pmax- t)|v|^{1/2}(t- q)\,\d t}{z - \dispersion(t)}\,f(q)\,\,\d q.
\end{equation}
Moreover, repeating the same arguments of \eqref{ishaq_inna} we have
\begin{equation}\label{q1_kuydurgi}
|(Q_1(z)\phi_0,\phi_1)_{L^2}| \le
\begin{cases}
\big\||v|^{1/2}\big\|_{L^2}\|f\|_{L^2} \big[\tilde c_1 - \tilde c_2\ln(z-\emax)\big] & \text{if $d=1,$} \\[1mm]
\tilde c_3\,\big\||v|^{1/2}\big\|_{L^2}\|f\|_{L^2}  & \text{if $d=2 $} \\[1mm]
\end{cases}
\end{equation}
for some $\tilde c_1,\tilde c_2,\tilde c_3>0.$
Thus,
\begin{equation}\label{def_ez}
\kappa_z:=(Q_1(z)\phi_0,\phi_0)= \int_{\T^d} \frac{|v(t-\pmax)|^2\,\d t}{z - \dispersion(t)}\in(0,+\infty)
\end{equation}
and
\begin{equation}\label{def_e0}
\kappa_1:=(Q_1(\emax)\phi_0,\phi_0):= \sup\limits_{z>\emax}\, (Q_1(z)\phi_0,\phi_0)= \int_{\T^d} \frac{|v(t-\pmax)|^2\,\d t}{\emax - \dispersion(t)}<+\infty.
\end{equation}

Note that since $0$ is not an isolated eigenvalue of $Q,$ hence, we cannot directly use  the perturbation theory for $\birman(z)$.

Let us show that if $z-\emax>0$ is sufficiently small, then the spectral projection
\begin{equation}\label{P_epsilon_z}
P_\epsilon(z):=\frac{1}{2\pi i}\int_{|\xi - \kappa_1|=\delta} \Big(\xi - QQ_1(z) + \frac{Q_1(z)}{a(z)^{1/2}[\kappa_z + \epsilon]^{1/2}}\Big)^{-1}d\xi
\end{equation}
is non-zero for any sufficiently small $|\epsilon|,\delta>0.$ Indeed, since $QQ_1(z)$ is rank-one and $(Q_1(z)\phi_0,\phi_0)$ is the unique positive eigenvalue of $QQ_1(z),$
$$
\Big(\xi - QQ_1(z)\Big)^{-1} = \frac{1}{\xi^2} + \frac{QQ_1(z)}{\xi(\xi - \kappa_z)},\qquad \xi\in\C\setminus\{0,\kappa_z\}.
$$
Therefore, for $\delta,|\epsilon|<\kappa_1/8$
\begin{align*}
P_\epsilon(z):=&\frac{1}{2\pi i}\int_{|\xi - \kappa_1|=\delta} \left[1-\Big(\xi - QQ_1(z)\Big)^{-1}\frac{Q_1(z)}{a(z)^{1/2}[\kappa_z + \epsilon]^{1/2}}\right]^{-1}\Big(\xi - QQ_1(z)\Big)^{-1} d\xi\\
=&\frac{1}{2\pi i}\int_{|\xi - \kappa_1|=\delta} \left[1-\left(\frac{1}{\xi^2} + \frac{QQ_1(z)}{\xi(\xi - \kappa_z)}\right)\frac{Q_1(z)}{a(z)^{1/2}[\kappa_z + \epsilon]^{1/2}}\right]^{-1}\left(\frac{1}{\xi^2} + \frac{QQ_1(z)}{\xi(\xi - \kappa_z)}\right) d\xi.
\end{align*}
Notice that
$$
\|QQ_1(z)f\|=\left|\int_{\T^d}\int_{\T^d}\frac{v(\pmax - t)|v|^{1/2}(t-q)\,\d t}{z-\dispersion(t)} \,f(q)\,\,\d q\right|\le \|f\|\,\,\||v|^{1/2}\|\,\int_{\T^d}\frac{|v(\pmax - t)|\,\,\d t}{z-\dispersion(t)}.
$$
Thus, as in the proof of \eqref{ishaq_inna},
\begin{equation*}
\frac{\|QQ_1(z)f\|}{\|f\|} \le  \||v|^{1/2}\|\,
\begin{cases}
\tilde c_4-\tilde c_5\ln(z-\emax) &\text{if $d=1,$}\\[2mm]
\tilde c_5 &\text{if $d=2$}
\end{cases}
\end{equation*}
for some $\tilde c_4,\tilde c_5>0$ independent of $z.$ Since $a(z)$ behaves as $(z-\emax)^{-1/2}$ for $d=1$ and as $-\ln(z-\emax)$ for $d=2,$ by \eqref{limit_q1_d1} and \eqref{limit_q1_d2} we obtain
\begin{align*}
&\left \|\left(\frac{1}{\xi^2} + \frac{QQ_1(z)}{\xi(\xi - \kappa_z)}\right)  \frac{Q_1(z)}{a(z)^{1/2}[\kappa_z + \epsilon]^{1/2}}\right\| \nonumber \\
&\hspace*{5cm}\le
\begin{cases}
-c_\delta (z-\emax)^{\gamma/4} \ln(z-\emax) &\text{if $d=1,$}\\[2mm]
-\frac{c_\delta}{\ln(z-\emax)} &\text{if $d=2$}
\end{cases}
\end{align*}
for some $c_\delta>0,$ where we took also account that $|\xi - \kappa_z|=\delta>0.$ In particular, for all sufficiently small $z-\emax>0$
\begin{align*}
P_\epsilon(z):=&\frac{1}{2\pi i}\int_{|\xi - \kappa_1|=\delta} \sum\limits_{n\ge0}  \left[\left(\frac{1}{\xi^2} + \frac{QQ_1(z)}{\xi(\xi - \kappa_z)}\right)\frac{Q_1(z)}{a(z)^{1/2}[\kappa_z + \epsilon]^{-1/2}}\right]^{n}\left(\frac{1}{\xi^2} + \frac{QQ_1(z)}{\xi(\xi - \kappa_z)}\right)d\xi,
\end{align*}
and thus,
\begin{align*}
P_\epsilon(z)= & \frac{1}{2\pi i}\int_{|\xi - \kappa_1|=\delta} \left(\frac{1}{\xi^2} + \frac{QQ_1(z)}{\xi(\xi - \kappa_z)}\right) \,d\xi\nonumber \\
+ & \sum\limits_{n\ge1} \frac{1}{2\pi i}\int_{|\xi - \kappa_1|=\delta} \left[\left(\frac{1}{\xi^2} + \frac{QQ_1(z)}{\xi(\xi - \kappa_z)}\right)\frac{Q_1(z)}{a(z)^{1/2}[\kappa_z + \epsilon]^{-1/2}}\right]^{n}\left(\frac{1}{\xi^2} + \frac{QQ_1(z)}{\xi(\xi - \kappa_z)}\right)\,d\xi\nonumber \\
=& \frac{QQ_1(z)}{\kappa_z} + \sum\limits_{n\ge1} \frac{(\kappa_z+ \epsilon)^{n/2}}{a(z)^{n/2}}\left[\frac{Q_1^nQQ_1}{\kappa_z^{2n+1}} + \sum\limits_{0\le k\le n-1} \frac{Q_1(z)^k QQ_1(z)^{n-k+1}}{\kappa_z^{2n+1}}\right],
\end{align*}
where in the second equality we used the Cauchy's Integral Theorem for analytic functions. Therefore,
\begin{equation}\label{projection_perturba}
P_\epsilon(z)\phi_0 = \phi_0 + \psi_z^\epsilon,
\end{equation}
where
\begin{equation}\label{def_psi_z_epsi}
\psi_z^\epsilon:=\sum\limits_{n\ge1} \frac{(\kappa_z+ \epsilon)^{n/2}}{a(z)^{n/2}}\left[\frac{Q_1(z)^n\phi_0 }{\kappa_z^{2n}} + \sum\limits_{0\le k\le n-1} \frac{ (Q_1(z)^{n-k+1}\phi_0,\phi_1)_{L^2}^{}\,Q_1(z)^k\phi_0}{\kappa_z^{2n+1}}\right].
\end{equation}
Now if $d=1,$ then by \eqref{q1f_phi1} and \eqref{limit_q1_d1} one has
\begin{align*}
\frac{(\kappa_z+ \epsilon)^{n/2}}{a(z)^{n/2}} \left\|\frac{Q_1(z)^n\phi_0 }{\kappa_z^{2n}} + \sum\limits_{0\le k\le n-1} \frac{ (Q_1(z)^{n-k+1}\phi_0,\phi_1)_{L^2}^{} \,Q_1(z)^k\phi_0 }{\kappa_z^{2n+1}}\right\|\qquad \qquad \\
\le - \tilde C n \,a(z)^{-\gamma/2}\ln(z-\emax)
\end{align*}
for some $\tilde C>0$ depending on $\kappa_1$ and $\delta.$  Hence, choosing $\delta_1>0$ such that $-\tilde C a(z)^{-\gamma/2}\ln(z-\emax)<\frac14$ for $z\in(\emax,\emax+\delta_1)$ we get
\begin{equation}\label{psiz_baho_d1}
\|\psi_z^\epsilon\|\le  \frac{-\tilde C\, a(z)^{-\gamma/2}\ln(z-\emax)}{(1- a(z)^{-\gamma/2})^2}< \frac{\|\phi_0\|}{2}.
 \end{equation}
If $d=2,$ then by \eqref{limit_q1_d2} $\|Q_1(z)\|\le C_2$ for any $z\ge\emax.$ Therefore, we can choose $\delta_1>0$ such that
\begin{equation}\label{psiz_baho_d2}
\|\psi_z^\epsilon\|\le  \frac{\tilde C\, a(z)^{-1/2}}{(1- a(z)^{-1/2})^2}< \frac{\|\phi_0\|}{2}
 \end{equation}
for any $z\in(\emax,\emax+\delta_1).$ The estimates \eqref{psiz_baho_d1} and \eqref{psiz_baho_d2} implies that for such $z,$ $P_\epsilon(z)\ne0.$ Let us show that there exists $\delta_2\in(0,\delta_1)$ such that for any $z\in(\emax,\emax+\delta_2)$ there exists a solution $\epsilon_z$ to the equation
$$
R(z,\epsilon)=0,
$$
where $R:[\emax,\emax+\delta_1)\times (-\kappa_1/8,\kappa_1/8) \to L^2(\T^d)$ is defined as
$$
R(z,\epsilon):=
\begin{cases}
P_\epsilon(z)\phi_0-\frac{QQ_1(z)P_\epsilon(z)\phi_0}{\kappa_z + \epsilon} - \frac{Q_1(z)P_\epsilon(z)\phi_0}{\sqrt{a(z)}(\kappa_z + \epsilon)^{1/2}} & \text{if $z\in(\emax,\emax+\delta_1),$}\\
\frac{\epsilon \phi_0}{\kappa_1 + \epsilon} & \text{if $z=\emax.$}
\end{cases}
$$
Indeed,  notice that
$$
\lim\limits_{z\to\emax} \|R(z,\epsilon) - R(\emax,\epsilon)\|_{L^2}=0
$$
so that $R\in C^0([\emax,\emax+\delta_1)\times (-\kappa_1/8,\kappa_1/8);L^2(\T^d))$ and
$$
\lim\limits_{z\to\emax} R(z,\epsilon)= R(\emax,0)=0.
$$
Moreover, since $R(z,\cdot)$ is analytic around $\epsilon=0$ and
$$
\frac{\partial R(z,\epsilon)}{\partial \epsilon}\Big|_{(z,\epsilon)= (\emax,0)} = \frac{\phi_0}{\kappa_1}\ne0,
$$
by the Implicit Function Theorem, there exists $\delta_2\in(0,\delta_1)$ such that for any $z\in[\emax,\delta_2)$ there exists a unique $\epsilon_z\in(-\kappa_1/8,\kappa_1/8)$ such that
$$
R(z,\epsilon_z) =0,\qquad z\in[\emax,\epsilon_z).
$$
Since $[\emax,\emax+\delta_1)\subset\R,$  the fact that $\emax$ is not an interior point does not affect, since we do not need that any regularity of the implicit function $\epsilon_z.$ However,  notice that $\epsilon_z\to0$ as $z\searrow \emax.$

Let us now introduce
$$
\chi(z):=\sqrt{a(z)}\big[\kappa_z + \epsilon_z\big],\qquad z\in(\emax,\emax+ \delta_2);
$$
then $\chi(z)>0$ and by the definition of $R(z,\epsilon)$ and $\epsilon_z,$ we have
\begin{equation*}
\frac{\chi^2}{a(z)}\,\phi_z = QQ_1(z) \phi_z + \frac{\chi(z)Q_1(z)\phi_z}{a(z)},\qquad z\in(\emax,\emax+\delta_2),
\end{equation*}
where $\phi_z:=P_{\epsilon_z}(z)\phi_0\ne0.$ Thus,
\begin{equation}\label{kernella_same}
\phi_z\in \Ker\left[1- \left(\frac{a(z)}{\chi(z)} + \frac{a(z)^2Q_1(z)}{\chi(z)^2}\right)\frac{Q_1(z)}{a(z)}\right] =\Ker\left[1 - \frac{a(z)Q}{\chi(z)} - \frac{Q_1(z)}{\chi(z)}\right],
\end{equation}
where in the equality we used \eqref{resolvent_Q}. By \eqref{qora_kozli_sanamjon}, this implies that $\phi_z$ is an eigenvector of $\birman(z)$ associated to its positive eigenvector $\chi(z).$ Recall that $\chi(z) \le \lambda(z).$

To establish $\chi(z)=\lambda(z)$ and $\lambda(z)$ is a simple eigenvalue of $\birman(z),$ we consider the self-adjoint operator
$$
\tilde \hamiltonian_\mu:=\hamiltonian_0 - \mu \potential;
$$
notice that $\tilde \hamiltonian_\mu$ differes from $\hamiltonian_\mu$ only with the sign of $\potential.$ Then associated Birman-Schwinger operator $\tilde \birman(z)$ satisfies $\tilde \birman(z) = - \birman(z).$ Since $\sum_x\hat v(x)=0,$ by the arguments above, there exists a positive eigenvalue $\tilde\xi(z)$ of $\tilde\birman(z)$ which satisfies
$$
\tilde \chi(z) = \sqrt{a(z)}(\kappa_z + o(1))\qquad \text{as $z\to\emax$.}
$$
Recall that $-\tilde \chi(z)$ is a negative eigenvalue of $\birman(z).$ Let us enumerate all nonzero eigenvalues $\eta_0(z),\eta_1(z),\ldots$ (counted with multiplicities) of $\birman(z)$ as  $|\eta_0(z)|\ge |\eta_1(z)|\ge 0$ and also all positive eigenvalues $\nu_0\ge \nu_1\ge \ldots$  of $|\birman(z)|.$ By \cite[Lemma 3.3]{GK:1969},
\begin{equation}\label{gihberg_lemma}
|\eta_0(z)\eta_1(z)\eta_2(z)| \le \nu_0(z)\nu_1(z)\nu_2(z).
\end{equation}
Note that by \eqref{eigenvalue_asymp_d1}-\eqref{eigenvalue_asymp_d2} applied to $|\birman(z)|,$ we have
\begin{equation*}
\frac{\nu_0(z)}{a(z)}= \sum\limits_{x\in\Z^d} |\hat v(x)| + o(1) \qquad\text{as $z\to\emax.$}
\end{equation*}
Moreover, by  \eqref{boje_moi},
\begin{equation}\label{asymp_nu12}
\sup\limits_{z>\emax}\,|\nu_i(z)| \le C_4 \sum\limits_{x\in\Z^d} |x|^{2-d+\gamma}|\hat v(x)|,\qquad i=1,2,
\end{equation}
and hence, if $\birman(z)$ has at least two positive eigenvalues with asymptotics $\ge \sqrt{a(z)}(c + o(1)),$ then recalling the definition of $-\tilde \chi(z)$ we get
$$
|\eta_0(z)\eta_1(z)\eta_2(z)|\ge a(z)^{3/2}(\tilde c + o(1))
$$
for some $\tilde c>0$ as $z\to+\infty.$  Hence, by \eqref{gihberg_lemma}
$$
\tilde c +o(1) \le \frac{\tilde C}{a(z)}\to0\qquad\text{as $z\to\emax,$}
$$
a contradiction. Thus,
$$
\lambda(z)=\sqrt{a(z)} \,[\kappa_1 + o(1)]\qquad \text{as $z\to\emax.$}
$$
Analogously,
$$
\tilde \lambda(z):=\sup\sigma(\tilde \birman(z)) =\sqrt{a(z)} \,[\kappa_1 + o(1)]\qquad\text{as $z\to\emax.$}
$$
Since  $|\eta_0\eta_1| = \lambda(z)\tilde \lambda(z)$ for small and positive $z-\emax,$  \eqref{gihberg_lemma} implies
$$
|\eta_n(z)|\le |\eta_2(z)| \le \frac{\nu_0(z)\nu_1(z)\nu_2(z)}{\lambda(z)\tilde\lambda(z)}
$$
for any $n\ge2.$ Now using the asymptotics of $\lambda(z),$ $\nu_0(z)$ and $\tilde \lambda(z)$ as well as estimate \eqref{asymp_nu12}, we find $\delta_3\in(0,\delta_2)$ such that
\begin{equation}\label{local_est}
\lambda_n(z) \le \tilde C_5:=\Big(\frac{C_4}{\kappa_1}\, \sum\limits_{x\in\Z^d} |\hat v(x)|\Big)^2 +1,\qquad n\ge1,
\end{equation}
for all $z\in(\emax,\emax+\delta_3),$ where $\lambda_0(z)=\lambda(z)\ge \lambda_1(z)\ge \ldots$ are all positive eigenvalues of $\birman(z)$ (counted with their multiplicities). Since $\lambda_n(z)\le \lambda_0(z) \le \|\birman(z)\|,$
\begin{equation}\label{local_inf_est}
\sup\limits_{z\ge \emax+\delta_3} \lambda_n(z) \le\|\birman(\emax+\delta_3)\|,\qquad n\ge0.
\end{equation}
Now \eqref{local_est} and \eqref{local_inf_est} implies  \eqref{jilmayasan_xayol_surasan} with $C_5:=\max\{\tilde C_5,\|\birman(\emax + \delta_3)\|\}.$

It remains to prove \eqref{issaqbaazi_he}. We set
\begin{equation}\label{resideue_lalalal}
\epsilon(z):= \frac{\lambda(z)^2}{a(z)} - \kappa_z,
\end{equation}
where $\kappa_z$ and $\kappa_1$ are given in \eqref{def_ez} and \eqref{def_e0}, respectively. By \eqref{estimate_ez_e0}, there exists $\delta_4\in(0,\delta_3)$ such that
\begin{equation}\label{ez_to_e0_estim}
\kappa_z = \kappa_1 + h_4(z),
\end{equation}
where $h_4\in C^0([\emax,\emax+\delta_1))$ with $h_4(z)=O\big((z-\emax)^{\gamma/4}\big)$ as $z\to\emax.$
Now we consider the spectral projection $P_\epsilon(z)$ in \eqref{P_epsilon_z} with $\epsilon:=\epsilon(z).$
By \eqref{projection_perturba}, $P_{\epsilon(z)}(z)\phi_0 = \phi_0 + \psi_z^{\epsilon(z)}.$ Since
$$
\frac{\lambda(z)^2}{a(z)}\,P_{\epsilon(z)}(z)\phi_0 = (QQ_1(z) + \frac{\lambda(z)Q_1(z)}{a(z)})\,P_{\epsilon(z)}(z)\phi_0
$$
(see \eqref{kernella_same}), we have
$$
\epsilon(z)\phi_0 =- [\kappa_z + \epsilon(z)]\psi_z^{\epsilon(z)}  + \phi_0 (Q_1(z)\psi_z^{\epsilon(z)},\phi_1)  + \frac{Q_1(z)\phi_0 + Q_1(z)\psi_z^{\epsilon(z)}}{a(z)^{1/2}}\,\big[\kappa_z + \epsilon(z)\big]^{1/2}.
$$
Multiplying this by $\phi_0$ and using the definition \eqref{def_psi_z_epsi} of $\psi_z^{\epsilon(z)}$ we get
$$
\begin{aligned}
\epsilon(z)\|\phi_0\|_{L^2}^2 = &- [\kappa_z + \epsilon(z)](\psi_z^{\epsilon(z)},\phi_0)_{L^2}  + \|\phi_0\|_{L^2}^2 (Q_1(z)\psi_z^{\epsilon(z)},\phi_1) \\
&+ \frac{(Q_1(z)\phi_0,\phi_0)_{L^2} + (Q_1(z)\psi_z^{\epsilon(z)},\phi_0)_{L^2}}{a(z)^{1/2}}\,\big[\kappa_z + \epsilon(z)\big]^{1/2}.
\end{aligned}
$$
Now if $d=1,$ then by \eqref{psiz_baho_d1}, \eqref{q1_kuydurgi}, \eqref{limit_q1_d1},
$$
|\epsilon(z)| \le \tilde C a(z)^{-\gamma/2}\ln^2(z-\emax),\qquad z\in(\emax,\emax+\delta_4).
$$
This, \eqref{ez_to_e0_estim} and \eqref{resideue_lalalal} implies \eqref{issaqbaazi_he}.

If $d=2,$ then  using the definition \eqref{def_psi_z_epsi} of $\psi_z^{\epsilon(z)},$ \eqref{ez_to_e0_estim} and \eqref{limit_q1_d2} we get
$$
(\psi_z^{\epsilon(z)},\phi_0)_{L^2} =a(z)^{-1/2}\Big[\frac{(Q_1(\emax)\phi_0,\phi_0)}{\kappa_1^{3/2}} + \frac{\|\phi_0\|^2(Q_1(\emax)^2\phi_0,\phi_0)}{\kappa_1^{5/2}} \Big] + a(z)^{-1/2} \,h_5(z),
$$
where $h_5\in C^0(\emax,\emax+\delta)$ with $h_5(z)=O((z-\emax)^{\gamma/4}).$ Moreover, by \eqref{ez_to_e0_estim}, \eqref{norm_converge_q1}, \eqref{q1f_phi1}, \eqref{def_psi_z_epsi}
$$
\begin{aligned}
(Q_1(z)\psi_z^{\epsilon(z)},\phi_1) = & a(z)^{-1/2}\Big[\frac{(Q_1(\emax)^2\phi_0,\phi_0)}{\kappa_1^{3/2}} + \frac{(Q_1(\emax)\phi_0,\phi_0)\,(Q_1(\emax)^2\phi_0,\phi_0)}{\kappa_1^{5/2}} \Big] \\
&+ a(z)^{-1/2} \,h_6(z),
\end{aligned}
$$
where $h_6\in C^0(\emax,\emax+\delta)$ with $h_6(z)=O((z-\emax)^{\gamma/4}).$
Moreover, again by \eqref{norm_converge_q1}
$$
\frac{(Q_1(z)\phi_0,\phi_0)_{L^2} + (Q_1(z)\psi_z^{\epsilon(z)},\phi_0)_{L^2}}{a(z)^{1/2}}\,\big[\kappa_z + \epsilon(z)\big]^{1/2}= \frac{\kappa_1^{1/2}(Q_1(\emax)\phi_0,\phi_0)}{a(z)^{1/2} } + \frac{h_7(z)}{a(z)^{1/2}},
$$
where $h_7\in C^0(\emax,\emax+\delta)$ with $h_7(z)=O((z-\emax)^{\gamma/4}).$ Hence,
\begin{equation}\label{kkkkka}
\epsilon(z) = \tilde C a(z)^{-1/2} + a(z)^{-1/2}h_8(z),
\end{equation}
where $h_8\in C^0(\emax,\emax+\delta)$ with $h_8(z)=O((z-\emax)^{\gamma/4}).$ Finally, since $a(z)^{1/2} (z-\emax)^{\gamma/4} = o((z-\emax)^{\gamma/8}),$ From \eqref{resideue_lalalal}, \eqref{ez_to_e0_estim} and \eqref{kkkkka}  we get \eqref{issaqbaazi_he}.
\end{proof}

Combining Proposition \ref{prop:asym_eigenvalues_bsh} and Lemma \ref{lem:b_sh_principle} we get

\begin{corollary}\label{cor:existence_for small_mu}
Assume Hypothesis \ref{hyp:maina} and let $C_5>1$ be given by Proposition \ref{prop:asym_eigenvalues_bsh}. Then for any $\mu\in(0,\frac{1}{C_5}):$
\begin{itemize}
\item[(a)] if $\sum\limits_{x\in\Z^d} \hat v(x)<0,$ then $\sigma_\disc(\hamiltonian_\mu) \cap (\emax,+\infty) =\emptyset;$

\item[(b)] if $\sum\limits_{x\in\Z^d} \hat v(x)\ge 0,$ then $\sigma_\disc(\hamiltonian_\mu) \cap (\emax,+\infty)$ is a singleton $\{E(\mu)\}.$ Moreover, the map $\mu\in(0,\frac{1}{C_5})\mapsto E(\mu)$ is analytic,  strictly increasing and $E(\mu)\to \emax$ as $\mu\to0$.
\end{itemize}
\end{corollary}

\begin{proof}
Let $\lambda(z)$ be given by \eqref{lambdaz}.

(a) If $\sum\limits_{x\in\Z^d} \hat v(x)<0,$ then by Proposition \ref{prop:asym_eigenvalues_bsh} (a)  for any $\mu<\frac{1}{C_5}$ and $z>\emax$ all positive eigenvalues of $\mu\birman(z)$ will be less than $1,$ i.e.,
$$
\Ker(1+\mu\birman(z))=0.
$$
Thus, by Lemma \ref{lem:b_sh_principle}, $\sigma_\disc(\hamiltonian_\mu)\cap (\emax,+\infty)=\emptyset.$

(b)  Assume that $\sum\limits_{x\in\Z^d} \hat v(x)\ge 0.$
Given $z>\emax,$ if $\lambda_0(z)=\lambda(z) \ge \lambda_1(z) \ge\ldots >0$ are all positive eigenvalues of $\birman(z),$ then by \eqref{jilmayasan_xayol_surasan}, $\mu\lambda_k(z)<1$ for any $\mu<\frac{1}{C_5}$ and $k\ge1.$ In particular, $\Ker (1+\mu\birman(z))$ is at most one-dimensional. By Proposition \ref{prop:asym_eigenvalues_bsh} (b), $\lambda(z) \to+\infty$ as $z\to \emax.$ Moreover, by Lemma \ref{lem:properties_of_lambda}, $\lambda(\cdot)$ is continuous and strictly decreasing in $(\emax,+\infty),$ and $\lim\limits_{z\to+\infty}\lambda(z)=0.$ Therefore,  for any $\mu\in(0,\frac{1}{C_5})$ there exists a unique $E(\mu)>\emax$ such that
$\mu \lambda(E(\mu)) =1.$ By Lemma \ref{lem:b_sh_principle}, $E(\mu)$ is the unique eigenvalue of $\hamiltonian_\mu$ in $(\emax,+\infty).$ By the Implicit Function Theorem in the monotonous case, the map $\mu\in(0,\frac{1}{C_5}) \mapsto E(\mu)$ is strictly increasing  and $E(\mu)\to \emax$ as $\mu\to0$.
\end{proof}

\section{Proofs of the main results}\label{sec:proof_of_main_results}

In this section we prove main results of the paper. 

\begin{proof}[Proof of Theorem \ref{teo:finiteness_eigenvalues}]
We show only \eqref{posi_eigenvalues}, and the proof of \eqref{nega_eigenvalues} is similar. Recall that  $\cN^+(\hat\hamiltonian_\mu,\emax) = \cN^+(\hamiltonian_\mu,\emax).$ Since
$$
\hamiltonian_\mu\le \hamiltonian_0 + \mu |\potential|,
$$
by the minmax principle, $\cN^+(\hamiltonian_\mu,\emax)\le \cN^+(\hamiltonian_0 + \mu |\potential|,\emax).$ Hence,
it suffices to establish
\begin{equation}\label{number_eigen_neg}
\cN^+(\hat\hamiltonian_0 + \mu |\hat\potential|,\emax) \le 1+ C_4\mu\,\sum\limits_{x\in\Z^d} |x|^{2-d+ \gamma }  |\hat v(x)|
\end{equation}
for $C_4$ of \eqref{boje_moi}.
Recall by \cite[Lemma 2.1 (iv)]{BPL:2018} that given $z>\emax,$
$$
\cN^+(\hamiltonian_0 + \mu |\potential|,z) = \cN^+(\mu|\birman(z)|,1).
$$
Let $\nu_0(z):=\big\||\birman(z)|\big\| \ge \nu_1(z)\ge\ldots>0$ be all positive eigenvalues of the $|\birman(z)|;$ since $\nu_k(z)\to0$ as $k\to\infty,$ there exists a unique  $k_\mu\ge0$ such that $\mu\nu_{k_\mu}\ge1$ and $\mu\nu_{k_\mu +1}<1.$ Therefore, by Lemma \ref{lem:asymp_lambda_sign_definite_case},
$$
\begin{aligned}
\cN^+(\hamiltonian_0+\mu |\potential|,z) = & \cN^+( \mu|\birman(z)|, 1)  =  1+ k_\mu \le 1 + \mu\sum\limits_{i=1}^{k_\mu} \nu_i(z) \\
\le  &1+ \mu \Big[\Tr(|\birman(z)|) - \nu_0(z)\Big] \le 1+ C_4 \mu \sum\limits_{x\in\Z^d} |x|^{2-d+ \gamma } |\hat v(x)|
\end{aligned}
$$
for any $z>\emax.$ Now letting $z\searrow \emax$ we get \eqref{number_eigen_neg}.
\end{proof}

\begin{proof}[Proof of Theorem \ref{teo:existence_and_nonexistence}]
Let $C_4(\dispersion,\hat v)$ and $C_4(-\dispersion,-\hat v)$ be given by Lemma  \ref{prop:asym_eigenvalues_bsh} applied with $(\dispersion,\hat v)$ and $(-\dispersion,-\hat v),$ respectively. Let $\mu_o:=\mu_0(\dispersion,\hat v):=\min\{1/C_4(\dispersion,\hat v),1/C_4(-\dispersion,-\hat v)\}>0.$  Now assertions of Theorem \ref{teo:existence_and_nonexistence} for small $\mu,$ i.e., for $\mu\in(0,\mu_o),$ follows from Corollary \ref{cor:existence_for small_mu} applied with $(\dispersion,\hat v)$ and $(-\dispersion,-\hat v),$ respectively.

Now we prove assertion (1) for all $\mu>0.$   Let
$$
E_0(\mu):=\|\hamiltonian_\mu\|:=\sup\sigma(\hamiltonian_\mu).
$$
Note that $E_0(\mu)=E(\mu)>\emax$ for $\mu\in(0,\mu_o),$ where $E(\mu)$ is the unique eigenvalue of $\hamiltonian_\mu$ given
by Corollary \eqref{cor:existence_for small_mu} (b).
By \eqref{essential_spectrum}
$$
E(\mu)= \sup\limits_{f\in L^2(\T^d),\,\|f\|=1}\,\, \max\{\emax, (\hamiltonian_\mu f,f)\}.
$$
Since the map $\mu\in(0,+\infty)\mapsto \min\{\emin, (\hamiltonian_\mu f,f)\}$ is nonincreasing for any $f\in L^2(\T^d)$ so is  $\mu\in(0,+\infty)\mapsto E_0(\mu).$ In particular, $E_0(\mu)>\emax$ for all $\mu>0.$ Thus, $E_0(\mu)\in\sigma_\disc(\hamiltonian_\mu)\cap(\emax,+\infty).$

The proof of assertion (2) follows from applying assertion (1) with $-\dispersion$ and $-\hat v,$ respectively.
\end{proof}

\begin{proof}[Proof of Theorem \ref{teo:asymptotics}]
We establish only the asymptotics of $E(\mu),$ i.e., in the case $\sum_x\hat v(x)\ge0,$ then the asymptotics of $e(\mu)$ follows by applying the established asymptotics  with $-\dispersion$ and $-\hat v.$  Let  $\mu_o>0$ be given by Theorem \ref{teo:existence_and_nonexistence}; recall that for any $\mu\in(0,\mu_o),$ $E(\mu)$ is the unique eigenvalue of $\hamiltonian_\mu$ in $(\emax,+\infty).$ By Lemma \ref{lem:b_sh_principle},
\begin{equation}\label{eigen_equaiaia}
\mu \lambda(E(\mu))= 1,\qquad \mu\in(0,\mu_o).
\end{equation}
We find the asymptotics of $E(\mu)$ using the asymptotics \eqref{dangal_dangal} and \eqref{issaqbaazi_he} of $\lambda(\cdot).$ Let us first establish the asymptotics of $a(z)$ as $z\to\emax.$ Since $\dispersion(\cdot)$ is has a unique non-degenerate maximum at $\pmax$ and $\dispersion$ is $C^{3,\alpha}$ around $\pmax,$ by the Morse Lemma there exists a neighborhood $U_{\pmax}\subset\T^d$ and a $C^{1,\alpha}$-diffeomorphism $\varphi: B_\gamma(0)\subset\R^d \mapsto U_{\pmax}$ such that $\varphi(0)=\pmax$ and
\begin{equation*}
\dispersion(\varphi(u)) = \emax - u^2,\qquad u\in B_\gamma(0).
\end{equation*}
Without loss of generality, we assume that $\alpha\in(0,\gamma/8].$
Writing
$$
a(z) = \int_{U_{\pmax}} \frac{\,\,\d q}{\dispersion(q) - z} + \int_{\T^d\setminus U_{\pmax}} \frac{\,\d q}{\dispersion(q) - z} =:I_1(z) +I_2(z),
$$
we observe that $I_2(\cdot)$ is analytic at $z=\emax.$ In $I_1(z)$ we make the change of variables  $q=\varphi(u):$
$$
\begin{aligned}
I_1(z) =& \int_{B_\gamma(0)} \frac{J(\varphi(u))du}{u^2 +z -\emax} \\
=&
J(\varphi(0))\int_{B_\gamma(0)} \frac{du}{u^2 + z -\emax} + \int_{B_\gamma(0)} \frac{[J(\varphi(u)) - J(\varphi(0))]du}{u^2 +z -\emax}\\
=&: I_{11}(z)+I_{12}(z),
\end{aligned}
$$
where $J\phi>0$ is the Jacobian of $\varphi.$ Since  $J\varphi \in C^{0,\alpha}(B_\gamma(0))$, there exists $c>0$ such that $|J\varphi(u) - J\varphi(0)| \le c|u|^\alpha$ for all $u\in B_\gamma(0),$  and hence, by Lemma \ref{lem:asymp_some_integral}
$$
|I_{12}(z)| \le cd\omega_d\int_0^\gamma \frac{r^{d+\alpha-1}dr}{r^2 + z -\emax} \le
\begin{cases}
2c(z - \emax)^{\frac{\alpha - 1}{2}} \int_0^{+\infty} \frac{r^\alpha dr}{r^2 + 1} & \text{if $d=1,$}\\[1mm]
\frac{2\pi c}{\gamma} + c_1 (z-\emax)^\alpha & \text{if $d=2,$}
\end{cases}
$$
where $c_2>0.$  Moreover, if $z>\emax,$
$$
I_{11}(z) =
\begin{cases}
\frac{\pi J(\varphi(0))}{(z - \emax)^{1/2}} \Big(1- \frac{2}{\pi}\,\arctan\frac{( z - \emax)^{1/2}}{\gamma} \Big) & \text{if $d=1,$}\\[3mm]
-\pi J(\varphi(0)) \ln(E(\mu) - \emax) \,\Big(1 - \frac{\ln (\gamma + z-\emax)}{\ln(z - \emax)}\Big) & \text{if $d=2.$}
\end{cases}
$$
Thus,
$$
\lim\limits_{z\searrow\emax} \frac{I_{12}(z)}{I_{11}(z)} =0
$$
and
\begin{equation}\label{az_asymptoti}
a(z) =
\begin{cases}
\frac{\pi J_0}{(z - \emax)^{1/2}} \Big(1+  (z-\emax)^\alpha h_1(z)\Big) & \text{if $d=1,$}\\[3mm]
-\pi J_0 \ln(z -\emax) \,\Big(1 + \frac{C_8}{\ln(z-\emax)} + \frac{(z-\emax)^\alpha  h_2(z)}{\ln(z-\emax)}\Big) & \text{if $d=2,$}
\end{cases}
\end{equation}
where $J_0:=J(\varphi(0)),$ $h_1, h_2\in C^0[\emax,\emax+\delta_0]$.

Assume that $\kappa_0:=\sum\limits_{x\in\Z^d} \hat v(x)>0.$ Then by \eqref{dangal_dangal} and \eqref{eigen_equaiaia}
\begin{equation}\label{eigenvalue_equationa}
\frac{1}{\mu a(E(\mu))} = \kappa_0 +
\begin{cases}
\frac{g_1(E(\mu))\,\ln (E(\mu)-\emax)}{\kappa_0 a(E(\mu))} & \text{if $d=1,$}\\[1mm]
\frac{C_6}{\kappa_0  a(E(\mu))} + \frac{(E(\mu) - \emax)^{\gamma/2}\,g_2(E(\mu))}{\kappa_0 a(E(\mu))} & \text{if $d=2,$}
\end{cases}
\end{equation}
where $C_6\in\R,$ $g_1,g_2\in C^0[\emax,\emax+\delta_0].$
Therefore, if $d=1,$ then by \eqref{az_asymptoti} and \eqref{eigenvalue_equationa} we get
\begin{equation}\label{implicit_equation_d1_vn0}
1 = \frac{\pi J_0\kappa_0 \mu}{(E(\mu) - \emax)^{1/2}} \Big(1+  (E(\mu)-\emax)^\alpha h_1(E(\mu))\Big) + \frac{\mu}{\kappa_0}\,g_1(E(\mu))\,\ln(E(\mu) - \emax)
\end{equation}
for any $\mu\in(0,\mu_o).$ Let $u_1(\mu)$ be such that
$$
(E(\mu) - \emax)^{1/2}= \pi J_0 \kappa_0 \mu\,(1 + u_1(\mu)).
$$
Then by  \eqref{implicit_equation_d1_vn0}
$$
|u_1(\mu)| \le \tilde c_1 \mu^\alpha,\qquad \mu\in(0,\mu_o),
$$
for some $\tilde c_1>0.$ Hence \eqref{E_ning_asymp1} in $d=1$ follows.

If $d=2,$ then by \eqref{az_asymptoti} and \eqref{eigenvalue_equationa}
\begin{align}\label{implicit_equation_d2_vn0}
1 =   -\pi J_0\kappa_0 \mu\, \ln(E(\mu) - \emax)\, & \Big(1+  \frac{h_2(E(\mu))}{\ln(E(\mu)-\emax)}\Big) \nonumber \\
& + \frac{C_8\mu}{\kappa_0} + \frac{\mu\,(E(\mu) - \emax)^{\gamma/2}\,g_2(E(\mu))}{\kappa_0}
\end{align}
for any $\mu\in(0,\mu_o).$ Let $u_2(\mu)$ be such that
$$
E(\mu) - \emax = e^{-\frac{1}{\pi J_0\kappa_0\mu}}\big(c + u_2(\mu)\big),\qquad c:=e^{\frac{C_6}{\pi J_0\kappa_0^2} - C_8}>0.
$$
Then \eqref{implicit_equation_d2_vn0} implies
$$
|u_2(\mu)|\le \tilde c_2\mu,\qquad \mu\in(0,\mu_o),
$$
for some $\tilde c_2>0.$ Hence \eqref{E_ning_asymp1} in $d=2$ follows.

Now assume that $\kappa_0:=\sum_x\hat v(x)=0.$ By Proposition \ref{prop:asym_eigenvalues_bsh}, $\kappa_1\in(0,+\infty).$ Moreover, by \eqref{eigen_equaiaia} and \eqref{issaqbaazi_he}
\begin{equation}\label{eigenvalue_equationb}
\frac{1}{\mu^2a(E(\mu))} =
\kappa_1 +
\begin{cases}
\frac{g_3(E(\mu))\ln^2(E(\mu)-\emax)}{a(E(\mu))^{\gamma/2}} & \text{if $d=1,$}\\[1mm]
\frac{C_7}{a(E(\mu))^{1/2}} + \frac{(E(\mu)-\emax)^{\gamma/2}g_4(z)}{a(E(\mu))^{1/2}} & \text{if $d=2,$}
\end{cases}
\end{equation}
where $C_7\in\R,$ $g_3,g_4\in C^0[\emax,\emax+\delta_0].$

Let $d=1.$ In this case by \eqref{eigenvalue_equationb} and  \eqref{az_asymptoti} we get
\begin{align}\label{implicit_equation_d1_v0}
& \frac{(E(\mu) - \emax)^{1/2}}{\pi J_0\kappa_0 \mu^2\Big(1+  (E(\mu)-\emax)^\alpha h_1(E(\mu))\Big)}\nonumber \nonumber \\
&\hspace*{4cm}=  \kappa_1 + \frac{(E(\mu)-\emax)^{\gamma/4} \,g_3(E(\mu))\,\ln^2(E(\mu) - \emax )}{\Big[\pi J_0 (1 + (E(\mu) - \emax)^\alpha \,h_1(E(\mu)))\Big]^{1/2}}
\end{align}
for any $\mu\in(0,\mu_o).$ Let $u_3(\mu)$ be such that
$$
(E(\mu) - \emax)^{1/2}= \pi J_0 \kappa_1 \mu^2\,(1 + u_1(\mu)).
$$
Then \eqref{implicit_equation_d1_v0} implies
$$
|u(\mu)| \le \tilde c_3 \mu^\gamma \ln^2\mu,\qquad \mu\in(0,\mu_o)
$$
for some $\tilde c_3>0.$ Hence \eqref{E_ning_asymp01} in $d=1$ follows.

Let $d=2.$ In this case by \eqref{eigenvalue_equationb} and  \eqref{az_asymptoti} we get
\begin{align}\label{implicit_equation_d2_v0}
&\frac{1}{-\pi J_0 \mu^2 \ln(E(\mu) -\emax) \,\Big(1 + \frac{C_8}{\ln(E(\mu)-\emax)} + \frac{(E(\mu)-\emax)^\alpha  h_2(z)}{\ln(E(\mu)-\emax)}\Big)} \nonumber \\
= &\kappa_1 +
\frac{C_7 + (E(\mu) - \emax)^{\gamma/2}g_4(E(\mu))}{\Big[-\pi J_0 \ln(E(\mu) -\emax) \,\Big(1 + \frac{C_8}{\ln(E(\mu)-\emax)} + \frac{(E(\mu)-\emax)^\alpha  h_2(z)}{\ln(E(\mu)-\emax)}\Big)\Big]^{1/2}}.
\end{align}
This equation can be rewritten as
$$
-\pi J_0 \kappa_1\mu^2 \ln(E(\mu) - \emax) +C_7\mu^2 [-\pi J_0  \ln(E(\mu) - \emax) ]^{1/2}\,[1 + o(1)] = 1.
$$
Note that the equation
$$
-\pi J_0 \kappa_1 \ln t + C_7[-\pi J_0  \ln t ]^{1/2} = \frac{1}{\mu^2}
$$
has a unique solution
$$
t = \exp\Big(- \frac{\big(\sqrt{4\kappa_1 + C_7^2\mu^2} - C_7\mu\big)^2}{4\pi J_0 \kappa_1^2\mu^2}\Big),
$$
hence if we set
$$
E(\mu) - \emax = \exp\Big(- \frac{\big(\sqrt{4\kappa_1 + C_7^2\mu^2} - C_7\mu\big)^2}{4\pi J_0 \kappa_1^2\mu^2}\Big)\,[c + u_4(\mu)],\qquad c:=e^{-\frac{C_8}{\pi J_0}}>0,
$$
for some $u_4(\mu)\in\R,$ then from \eqref{implicit_equation_d2_v0} we get
$$
|u_4(\mu)| \le \tilde c_4 \mu,\qquad \mu\in(0,\mu_o),
$$
for some $\tilde c_4>0.$ Hence \eqref{E_ning_asymp01} in $d=2$ follows.
\end{proof}

\appendix

\section{Asymptotics of some parametric integrals}

In this paper we frequently use the following technical tool.

\begin{lemma}\label{lem:asymp_some_integral}
Given $\alpha\ge0$ and $r_0\in(0,1),$ consider the integral
\begin{equation}\label{def:t_d}
T_\alpha(\omega):= \int_0^{r_0} \frac{r^\alpha dr}{r^2 + \omega^2},\qquad \omega>0.
\end{equation}
Then for any $\alpha\ge0$ there exist a polynomial $P_\alpha(\omega)$ such that $P_\alpha\equiv0$ for $\alpha\in[0,1],$  $P_\alpha\equiv T_\alpha(0):=\frac{r_0^{\alpha-1}}{\alpha-1}>0$ for $\alpha\in(1,2]$ and $P_\alpha$ is at most of order  $[\alpha]-2$ if $\alpha>2,$ where $[\alpha]$ is the integer part of $\alpha,$ and $g_\alpha\in L^\infty(0,+\infty)$  such that

\begin{itemize}
\item[(a)] for  $\alpha\in[0,1]:$
$$
T_\alpha(\omega)=
\begin{cases}
\frac{\pi}{2\omega}\Big[1 - \frac{2}{\pi}\,\arctan\frac{\omega}{r_0}\Big] & \text{if $\alpha=0,$}\\[2mm]
\frac{g_\alpha(\omega)}{\omega^{1-\alpha}} & \text{if $\alpha\in(0,1),$}\\[2mm]
- \ln\omega\,\Big[1 +  \frac{\ln(r_0^2 + \omega^2)}{2\ln\omega}\Big] & \text{if $\alpha=1;$}
\end{cases}
$$

\item[(b)] Let $\alpha>1.$ Then
$$
T_\alpha(\omega) \le T_\alpha(0)
$$
and
$$
T_\alpha(\omega)=  P_\alpha(\omega) +
g_\alpha(\omega)\,\omega^{[\alpha]-1}.
$$
\end{itemize}

\end{lemma}

\begin{proof}
(a) The asymptotics of $T_\alpha$ for $\alpha\in\{0,1\}$ is clear. In this case we define $g_\alpha\equiv1.$ If $\alpha\in(0,1),$ then using the change of variables $r=wt$ in the integral we get
$$
T_\alpha(\omega) \le \omega^{\alpha-1} \int_0^{+\infty} \frac{t^\alpha \,\d t}{t^2+1}.
$$
Hence,
$$
g_\alpha(\omega): = \omega^{1 - \alpha} T_\alpha(\omega)
$$
satisfies $\|g_\alpha\|_{L^\infty(0,+\infty)} \le \int_0^{+\infty} \frac{t^\alpha \,\d t}{t^2+1}<+\infty.$

(b) Let $\alpha\in(1,2].$ Then using the change of variable $r=t^{\frac{1}{\alpha-1}}$ in $T_\alpha$ we get
$$
T_\alpha(\omega) = \frac{1}{\alpha-1}\int_0^{r_0^{\alpha-1}} \frac{t^{\frac{2}{\alpha-1}} \,\d t}{t^{\frac{2}{\alpha-1}} + \omega^2} = \frac{r_0^{\alpha-1}}{\alpha-1} + \frac{\omega^2}{\alpha-1} \int_0^{r_0^{\alpha-1}} \frac{\,\d t}{t^{\frac{2}{\alpha-1}} + \omega^2}.
$$
Now using the change of variable $t=\omega^{\alpha-1}s$ we get
$$
\int_0^{r_0^{\alpha-1}} \frac{\,\d t}{t^{\frac{2}{\alpha-1}} + \omega^2} \le \omega^{\alpha-1} \int_0^{+\infty} \frac{ds}{s^{\frac{2}{\alpha-1}} +1}
$$
so that
$$
g_\alpha(\omega):= \omega^{1-\alpha} \big[T_\alpha(\omega) - T_\alpha(0)\big]
$$
satisfies $\|g_\alpha\|_{L^{\infty}(0,+\infty)} \le \frac{1}{\alpha-1} \int_0^{+\infty} \frac{ds}{s^{\frac{2}{\alpha-1}} +1} <+\infty.$

Let $n:=[\alpha]\ge2.$ Note that if $n$ is even, then
$$
r^n +(-1)^{[n/4]}\omega^n = (r^2+\omega^2)(r^{n-2} - r^{n-4}\omega^{2} +\ldots + (-1)^{[n/4]}\omega^{n-2}).
$$
Thus
$$
P_\alpha(\omega):= \int_0^{r_0} \Big((r^{n-2} - r^{n-4}\omega^{2} +\ldots + (-1)^{[n/4]}\omega^{n-2})\Big) r^{\alpha-n} dr;
$$
is a polynomial of order $n-2$ and
$$
T_\alpha(\omega) = P_\alpha(\omega) - (-1)^{[n/4]} \omega^{n} \int_0^{r_0} \frac{r^{\alpha -n} dr}{r^2 + \omega^2}.
$$
Note that
$$
\int_0^{r_0} \frac{r^{\alpha -n} dr}{r^2 + \omega^2} \le \omega^{\alpha-n-1} \int_0^{+\infty} \frac{r^{\alpha-n}dr}{r^2+1},
$$
where the last integral is finite since $\alpha-n\in[0,1).$ Hence,
$$
g_\alpha(\omega): = \omega^{1-\alpha} \big[T_\alpha(\omega) - P_{\alpha}(\omega)\big]
$$
satisfies
$\|g_\alpha\|_{L^\infty(0,+\infty)} \le \int_0^{+\infty} \frac{r^{\alpha-n}dr}{r^2+1}<+\infty.$

If $n\ge3$ is odd, then
$$
r^{n-1} +(-1)^{[\frac{n-1}{4}]}\omega^{n-1} = (r^2+\omega^2)(r^{n-3} - r^{n-5}\omega^{2} +\ldots + (-1)^{[n/4]}\omega^{n-3}).
$$
Thus,
$$
P_\alpha(\omega):= \int_0^{r_0} \Big(r^{n-3} - r^{n-5}\omega^{2} +\ldots + (-1)^{[n/4]}\omega^{n-3}\Big)r^{1+\alpha-n} dr
$$
is a polynomial of order $n-3$ and
$$
T_\alpha(\omega) = \tilde P_{n-3}(\omega) -(-1)^{[\frac{n-1}{4}]}\omega^{n-1} \int_0^{r_0} \frac{r^{1+\alpha-n}dr}{r^2 + \omega^2},
$$
Now as in  the case of $\alpha\in(1,2),$
$$
\int_0^{r_0} \frac{r^{1+\alpha-n}dr}{r^2 + \omega^2} \le \frac{\omega^{\alpha-n}}{\alpha-n} \int_0^{+\infty} \frac{dr}{r^{\frac{2}{\alpha -n}} +1},
$$
therefore,
$$
g_\alpha(\omega):= \omega^{1-\alpha}\big[ T_\alpha(\omega) - P_{n-3}(\omega)\big]
$$
satisfies
$\|g_\alpha\|_{L^{\infty}(0,+\infty)}\le \frac{1}{\alpha-n} \int_0^{+\infty} \frac{dr}{r^{\frac{2}{\alpha -n}} +1}<+\infty.$
\end{proof}

\end{document}